\documentclass[11pt]{article}
\newcommand{\ignore}[1]{}

\usepackage[normalem]{ulem}

\usepackage{amsfonts,amsmath,amsthm,times,fullpage,amssymb,epsfig,graphicx,enumerate,bm}

\usepackage{hyperref}

\hypersetup{colorlinks=false}

\usepackage[dvipsnames]{xcolor}

\usepackage{color}
\usepackage{dsfont}
\usepackage{wrapfig}
\usepackage{relsize}

\usepackage{url} 

\usepackage[left=1.0in,top=1.0in,right=1.0in,bottom=1.0in,nohead,textheight=10in,footskip=0.3in]{geometry}

\usepackage[noend]{algpseudocode} 
\usepackage[nothing]{algorithm}  
\usepackage{caption}

\newtheorem{theorem}{Theorem}[section]
\newtheorem{lemma}[theorem]{Lemma}
\newtheorem{corollary}[theorem]{Corollary}

\newtheorem{definition}[theorem]{Definition} 
\newtheorem{remark}{Remark}[section]
\newtheorem{proposition}[theorem]{Proposition}

\newcommand{\blue}[1]{{\color{black}{#1}}}

\newcommand\ex{{\mathbb{E}}}

\newcommand{\beq}{\begin{equation}}
\newcommand{\eeq}{\end{equation}}

\newcounter{fooTH}
\newcounter{fooEQ}

\def\poly{{\rm poly}}
\providecommand{\keywords}[1]
{
  \small	
  \textbf{\textit{Keywords---}} #1
}

\begin{document}

\title{Efficiently list-edge coloring multigraphs asymptotically optimally}

\author{
Fotis Iliopoulos
\thanks{Research supported by NSF grant CCF-1514434 and the Onassis Foundation.} \\ 
University of California Berkeley \\
{\small fotis.iliopoulos@berkeley.edu}
 \\
 \and 
Alistair Sinclair
\thanks{Research supported by NSF grants CCF-1514434 and CCF-1815328.} \\ 
University of California Berkeley \\
{\small sinclair@berkeley.edu}
}

\date{\empty}

\maketitle

\begin{abstract}

We give polynomial time algorithms for the seminal results of Kahn~\cite{kahnChrom,kahnListChrom},
who showed that the Goldberg-Seymour and List-Coloring conjectures for (list-)edge coloring 
multigraphs hold asymptotically.
Kahn's arguments are based on the probabilistic method and are non-constructive. Our key insight is 
that we can combine sophisticated techniques due to Achlioptas, Iliopoulos and 
Kolmogorov~\cite{AIK} for the analysis of 
local search algorithms with correlation decay properties of the probability spaces on matchings used
by Kahn in order to construct efficient edge-coloring algorithms.

\end{abstract}


\setcounter{page}{1}\maketitle

\newpage
\section{Introduction}\label{sec:intro}

In graph edge coloring one is given a (multi)graph $G=(V,E)$ and the goal is to find an assignment of one of $q$ colors to each edge $e\in E$ so that no pair of adjacent edges share the same color. The \emph{chromatic index}, $\chi_e(G)$, of $G$ is the smallest integer $q$ for which this is possible. In the more general \emph{list-edge coloring} problem, a  list of $q$ allowed colors is specified for each edge. A graph is $q$-list-edge colorable if it has a list-coloring no matter how the lists are assigned to  each edge. The \emph{list chromatic index}, $\chi^{\ell}_e(G) $, is the smallest $q$ for which $G$ is $q$-list-edge colorable.
 
Edge coloring is one of the most fundamental and well-studied coloring problems  with various applications in computer science (e.g., ~\cite{chen2011approximating,feige2002approximating, kahn1996asymptotically, kahnChrom,kahnListChrom, nishizeki19901, plantholt1999sublinear, sanders2008asymptotic,scheide200715,scheide2010graph,vizing1964estimate}).  To give just one representative example,  if edges represent data packets then an edge coloring with $q$ colors specifies a schedule for exchanging the packets directly and without node contention. In this paper we are interested in designing algorithms for efficiently edge coloring and list-edge coloring multigraphs.  To formally describe our results, we need some notation.

For a multigraph $G$, let $\mathcal{M}(G)$ denote the set of matchings of $G$. A \emph{fractional edge coloring} is a set  $\{M_1, \ldots, M_{\ell} \} $ of matchings and corresponding positive real weights $\{w_1, \ldots, w_{\ell} \}$, such that the sum of the weights of the matchings containing  each edge is one, i.e., $\forall e \in E$, $ \sum_{ M_i: e \in M_i } w_i =1$.  A fractional edge coloring is a  \emph{fractional edge $c$-coloring}  if $\sum_{ M \in \mathcal{M}(G) } w_M = c $. The \emph{fractional chromatic index} of $G$, denoted by $\chi_e^*(G)$, is the minimum $c$ such that $G$ has a fractional edge $c$-coloring.
 
Let $\Delta = \Delta(G)$ be the maximum degree of $G$ and define  $\Gamma:= \max_{H \subseteq V, |H| \ge 2 } \frac{|E(H)|  }{ \lfloor |H|/2  \rfloor}$, where $E(H)$ is the set of edges of the induced subgraph $H$. Both of these quantities are obvious lower bounds for the chromatic index and it is known~\cite{edmonds1965maximum} that $\chi_e^*(G) = \max (\Delta, \Gamma)  $. Furthermore,  Padberg and Rao~\cite{padberg1982odd} show that the fractional chromatic index of a multigraph, and indeed an optimal fractional edge coloring, can be computed in polynomial time. 

Goldberg and Seymour independently stated the now  famous conjecture  that every multigraph $G$ satisfies $\chi_e(G) \le \max \left( \Delta+1, \lceil \chi_e^*(G) \rceil \right)$. In a seminal paper~\cite{kahnChrom}, Kahn showed that the Goldberg-Seymour conjecture holds asymptotically:
\begin{theorem}[\cite{kahnChrom}]\label{GS_kahn}
The chromatic index of a multigraph~$G$ satisfies
%\marginpar{\tiny AS: Reworded theorem.   CHECK!!!}
%For any multigraph~$G$, we have
$\chi_e(G) \le (1 +o(1)) \chi_e^*(G).$
%\begin{align*}
%\chi_e(G) \le (1 +o(1)) \chi_e^*(G).
%\end{align*}
\end{theorem}
\noindent (Here $o(1)$ denotes a term that tends to zero as $\chi_e(G)\to\infty$.)
%\AS{Added parenthesis explaining $o(1)$ term. \red{OK}}
Later Kahn proved  the analogous result for list-edge coloring~\cite{kahnListChrom}, establishing that the List Coloring Conjecture, which asserts that  $\chi_e^{\ell}(G)  = \chi_e(G)$ for any multigraph $G$, also holds asymptotically:
\begin{theorem}[\cite{kahnListChrom}]\label{list_kahn}
The list chromatic index of a multigraph~$G$ satisfies
%For any multigraphs~$G$, we have
$\chi_e^{\ell}(G) \le (1 +o(1)) \chi_e^*(G).$
%\marginpar{\tiny AS: Reworded theorem.   CHECK!!!}
\end{theorem}

The proofs of Kahn use the probabilistic method and are not constructive. The main contribution of this paper is to provide polynomial time algorithms for the above results, as follows:
\begin{theorem}\label{main}
There exists a randomized algorithm that, given a multigraph $G$ on $n$ vertices, constructs  a  $(1 +o(1)) \chi_e^*(G)$-edge coloring of $G$ in expected polynomial time.
\end{theorem}

\begin{theorem}\label{main_list}
There exists a randomized  algorithm that, given a multigraph $G$ on $n$ vertices and an
arbitrary list of $(1 +o(1)) \chi_e^*(G)$ colors for each edge, constructs a valid list-edge coloring of $G$   in expected polynomial time.
\end{theorem}

Clearly, Theorem~\ref{main_list} subsumes Theorem~\ref{main}. Moreover, in a very recent breakthrough~\cite{chen2019proof},  Chen, Jing and Zang proved the (non-asymptotic) Goldberg-Seymour conjecture without exploiting the arguments of Kahn.
Even before this work, the results of Sanders and Steurer~\cite{sanders2008asymptotic}  and Scheide~\cite{scheide2010graph} already give deterministic polynomial time algorithms for  edge coloring multigraphs asymptotically optimally, again without exploiting the arguments of Kahn. Nonetheless, we choose to present the proof of Theorem~\ref{main} for three reasons. First and most importantly,    its proof is  significantly easier than that of Theorem~\ref{main_list}, while it   contains many of  the key ideas required for proving Theorem~\ref{main_list}.   Second, our algorithms and techniques are very different from those of~\cite{chen2019proof,sanders2008asymptotic,scheide2010graph}. Finally, we show that the algorithm of Theorem~\ref{main} is \emph{commutative}, a notion introduced by Kolmogorov~\cite{Kolmofocs}. This fact may be of independent interest  as we discuss in Remark~\ref{comm_remark} in Section~\ref{AlgorithmicLLLCriterion}.

To the best of our knowledge, Theorem~\ref{main_list} is the first result  to give an asymptotically optimal polynomial time algorithm for list-edge coloring multigraphs.

\subsection{Technical Overview}

The proofs of Theorems~\ref{GS_kahn} and~\ref{list_kahn}  are based on a very sophisticated variation of what is known as the {\it semi-random method\/} (also known as the ``naive coloring procedure"), which is the main technical tool behind some of the strongest graph coloring results, e.g.,~\cite{JO,kahn1996asymptotically,kim1995brooks,molloy1998bound}. The idea is to gradually color the graph in iterations, until we reach a point where we can finish the coloring using a  greedy algorithm. In its most basic form, each iteration consists of the following  simple procedure:
%\marginpar{\tiny AS: Switched vertices to edges: CHECK this still reads OK}
%(using vertex coloring  as a canonical example): 
assign to each edge  a color chosen uniformly at random; then uncolor any edge which receives the same color as one of its neighbors. Using the  Lov\'{a}sz Local Lemma (LLL)~\cite{LLL} and concentration inequalities, one typically shows that, with positive probability,  the resulting partial proper coloring has useful  properties that allow for the continuation of the argument in the next iteration. For a nice exposition of both the method and the proofs of Theorems~\ref{GS_kahn} and~\ref{list_kahn}, the reader is referred to~\cite{mike_book}.

The key new  ingredient in Kahn's arguments is the method of assigning colors to edges. For each color $c$,  we choose a matching~$M_c$ from some \emph{hard-core} distribution on $\mathcal{M}(G)$ and assign the color $c$ to the edges in~$M_c$.  The idea is that, by assigning each color exclusively to the edges of one matching, we avoid conflicting color assignments and the resulting uncolorings.

The existence of such hard-core distributions is guaranteed by  the characterization of the matching polytope due to Edmonds~\cite{edmonds1965maximum}   and a result by Lee~\cite{lee1990some} (also shown independently by Rabinovich et al.~\cite{rabinovich1992quadratic}). The crucial fact about them  is that they are  endowed with very useful approximate stochastic independence properties, as  was shown by Kahn and Kayll in~\cite{kahn1997stochastic}. In particular, for every edge $e$, conditioning on events that are determined by edges far enough from~$e$ in the graph does not effectively alter the probability of~$e$ being in the matching.

The reason why this property is important is because it enables the application of a sophisticated version of what is known as the \emph{Lopsided}  Lov\'{a}sz Local Lemma. Recall that the original statement  of the LLL asserts, roughly, that, given a family of ``bad" events in a probability space, if each bad event individually is not very likely and, in addition, is independent of all but a small number of other bad events, then the probability of avoiding all bad events is strictly positive.  The Lopsided LLL  used by Kahn 
generalizes this criterion as follows. For each bad event $B$, we fix a positive real number $\mu_B$ and require that conditioning on all but a small number of other bad events doesn't make the probability of $B$ larger than $\mu_B$. Then, provided the~$\mu_B$ are small enough,  the conclusion of the LLL still holds. In other words, one replaces the ``probability of a bad event" in the original LLL statement with the ``boosted"  probability of the event, and  the notion of ``independence" by the notion of ``sufficiently mild negative correlation".

Notably, the breakthrough result of Moser and Tardos~\cite{Moser,MT} that made the LLL constructive for the vast majority of its applications does not apply in this case, mainly for two reasons.  First, the algorithm of Moser and Tardos applies only when the underlying probability measure of the LLL application is a \emph{product} over explicitly presented variables. Second, it relies on a particular type of  dependency (defined by shared variables). The lack of an efficient algorithm for Lopsided LLL applications is the primary obstacle to making the arguments of Kahn constructive. 

Our main technical contribution is  the design and analysis  of such  algorithms. Towards this goal, we use the flaws-actions framework introduced in~\cite{AIJACM} and further developed in~\cite{AIK,AIS,StochControl,HIK,LLLWTL}. In particular, we  use the algorithmic LLL criterion  for the analysis of stochastic local search algorithms developed by Achlioptas, Iliopoulos and Kolmogorov in~\cite{AIK}. We start by showing that there is a connection between this criterion and the version of the Lopsided
%\AS{Should it be "Lopsided version of the LLL" or "version of the Lopsided LLL"?  \\ \red{ I am not entirely sure I understand the difference between the two but, presumably, the first one suggests that there is one version of the Lopsided LLL. If that's the case, then we should go with second one.} \\ \green{Right, if we want to say that there are multiple versions of the Lopsi LLL, then we should say "version of the lopsided LLL used by Kahn".  Both here and below.}  }
LLL used by Kahn, in the sense that the former can be seen as  the constructive counterpart of the latter.  However, this observation by itself is not sufficient, since the result of~\cite{AIK} is a tool for analyzing a \emph{given} stochastic local search algorithm. Thus, we are still left with the task of designing the algorithm before using it. Nonetheless, this connection provides valuable intuition on how to realize this  task. Moreover, we believe it is of independent interest as it provides an explanation for the success of various algorithms (such as~\cite{molloy2017list}) inspired by the techniques of Moser and Tardos, which were not  tied to a known form of the LLL.

To get a feeling for the nature of our algorithms, it is helpful  to have some intuition for the criterion of~\cite{AIK}. There, the  input  is the algorithm to be analyzed and a probability measure  $\mu$ over the state space of the algorithm. The goal of the algorithm is to reach a state that avoids a family of bad subsets of the space which we call {\it flaws}. It does this by focusing on a flaw that is currently present at each step, and taking  a (possibly randomized) action to address it. At a high level, the role of the measure is to gauge how efficiently the algorithm rids the state of flaws, by quantifying the trade-off between the probability that a flaw is present at some inner state of the execution of the algorithm and the number of other flaws each flaw can possibly introduce when the algorithm addresses it. In particular, the quality of the convergence criterion is affected by the \emph{compatibility} between the measure and the algorithm.

Roughly, the states of our algorithm will be matchings in a multigraph  (corresponding to color classes) and the goal will be to construct matchings that avoid certain flaws. To that end, our algorithm will locally modify each flawed matching by (re)sampling matchings in subgraphs of $G$ according to distributions induced by the hard-core distributions used in Kahn's proof. The fact that correlations decay with distance in these distributions allows us to prove that, while the changes are local, and hence not many new flaws are introduced at each step, the compatibility of our algorithms with these hard-core distributions is high enough to allow us to successfully apply the criterion of~\cite{AIK}.

%It is a also long and complicated, so we will not present it here in full detail and we will only provide  a very high level  sketch instead following the exposition in~\cite{mike_book}. However,  we will state formally every originally non-constructive statement for which we provide an algorithm.

%\subsection{Related Work}
%
%
%Vizing's algorithm can be generalized to color multigraphs with $\Delta + \mu$ colors, where $\mu$ is the multiplicity of an edge.  For example,   Steurer and Sanders $(1+\epsilon) \chi' + O(1/\epsilon)$. GS for specific families

\subsection{Organization of the Paper} 
In Section~\ref{Background} we present the necessary background. In Section~\ref{connection} we show a useful connection between the version of the Lopsided LLL used by Kahn and the algorithmic
%\AS{See previous margin note: should be consistent with that! \red{OK. I'll change it if you agree with the above comment. } \\ \green{See above. So this one is OK as is, right?}}
LLL criterion of~\cite{AIK}. In Section~\ref{main_1}  we present the proof of Theorem~\ref{main}.  In Section~\ref{list_main}, we sketch the proof of Theorem~\ref{list_kahn}  and then prove Theorem~\ref{main_list}.

\section{Background and Preliminaries}\label{Background}

\subsection{The Lopsided Lov\'{a}sz Local Lemma}\label{lopsi_LLLL}

 Erd\H{os} and Spencer~\cite{LopsTrav} noted that independence in the LLL can be replaced by positive correlation, yielding the original version of what is known as  the \emph{Lopsided} LLL,  more sophisticated  versions of which have also been established in~\cite{albert1995multicoloured,dudek2012rainbow}. Below we state   the Lopsided LLL in one of its most powerful forms.
	
\begin{theorem}[General Lopsided LLL]\label{generalLLLL}
Let $(\Omega,\mu)$ be a probability space and $\mathcal{B} = \{B_1, B_2,\ldots,B_m\}$ be a set of $m$ (bad) events. For each $i \in [m]$, let $L(i) \subseteq	 [m] \setminus \{i\}$ be such that $\mu(B_i \mid \bigcap_{j \in S} \overline{B_j}) \le  b_i$ for every $S \subseteq [m] \setminus (L(i) \cup \{i\})$. If there exist positive real numbers $\{x_i\}_{i=1}^m$ such that
\begin{equation}\label{eq:LLL}
b_i \le x_i \prod_{j \in L(i)} (1-x_j)  \enspace \text{  for all $i \in [m]$}  , 
\end{equation}
then the probability that none of the events in $\mathcal{B}$ occurs is at least $\prod_{i=1}^m (1- x_i)> 0$. 
\end{theorem}
 Note  that in most applications of the Lopsided LLL the definition of sets $\{ L(i) \}_{i \in [m] } $  is ``symmetric", in the sense that if $j \in L(i)$ then $i \in L(j)$ for every $i,j \in [m]$. With that in mind, any (undirected) graph on $[m]$ that includes every edge $(i,j)$ such that $j \in L(i) $ or $i \in L(j) $ is called a \emph{lopsidependency} graph.

\subsection{An Algorithmic LLL Criterion.}\label{AlgorithmicLLLCriterion} 

Let $\Omega$ be a discrete state space, and let
$F = \{f_1, f_2, \ldots, f_m \}$ be a collection of subsets (which we call {\it flaws\/}) of $\Omega$. We define  
$\bigcup_{i \in [m]} f_i = \Omega^*$. Our goal is to find a state 
$\sigma\in\Omega\setminus \Omega^*$; we refer to such states as {\it flawless}.
%Since we will be interested in algorithms that search for perfect objects, we sometimes refer to $\Omega$ as state space and to its elements as states.

For a state $\sigma$, we denote by $U(\sigma) = \{ f_j \in F \text{ s.t. } f_j \ni \sigma \}$ the set
of flaws present in~$\sigma$. We consider local search algorithms working on~$\Omega$ which, 
in each flawed state~$\sigma\in\Omega^*$, choose a flaw~$f_i$ in $U(\sigma)$ and randomly
move to a nearby state in an effort to fix~$f_i$.
We will assume that, for every flaw $f_i$ and every state $\sigma \in f_i$, there is a non-empty 
set of {\it actions\/} $a(i,\sigma) \subseteq \Omega$ such that \emph{addressing} flaw $f_i$
in state $\sigma$ amounts to selecting the next state~$\tau$ from $a(i,\sigma)$ according
to some probability distribution $\rho_i(\sigma, \tau)$. 
%We call $a(i,\sigma)$ the set of \emph{actions} for addressing flaw $f_i$ at $\sigma$
Note that potentially $a(i,\sigma) \cap f_i \ne \emptyset$, i.e., addressing a flaw does not
necessarily imply removing it.   We write $\sigma \xrightarrow{i} \tau$ to denote the fact
%\marginpar{\tiny AS: Should this be ``event" or ``fact"?}
that the algorithm  addresses flaw $f_i$ at $\sigma$ and moves to $\tau$.

Throughout the paper we consider algorithms that start from a state $\sigma\in\Omega$
picked from an initial distribution~$\theta$, and then repeatedly pick a flaw that is present
in the current state and address it.  The algorithm always terminates when it encounters a
%\AS{Turned the two items below into definitions.  OK?  Also, removed ``potential" from causality here and throughout the paper \red{OK}}
flawless state.

%\begin{causality}
\begin{definition}[Causality]
We say that flaw $f_i$  \emph{causes}  $f_j$ if there exists a  transition $\sigma \xrightarrow{i}  \tau$  such that (i)~$f_j \ni \tau$; (ii) either  $f_i = f_j$ or $f_j \not\ni \sigma$. 
\end{definition}
%\end{causality}
%\begin{cd}
\blue{
\begin{definition}[Causality Graph]
Any (undirected) graph $C=C(\Omega,F)$ on $[m]$ that includes every edge  $(i ,j)$ such that either $f_i$  causes $f_j$ or $f_j$ causes $f_i$ is called a  \emph{causality graph}. We write 
$\Gamma(i)$ for the set of neighbors of~$i$ in this graph. We also write $i\sim j$ to denote that $j \in \Gamma(i)$ (or equivalently, $i \in \Gamma(j)$). 
\end{definition}
}
%\end{cd}

%Throughout this paper we consider only bidirected causality digraphs, in which there is an edge from $f_i$ to $f_j$ if and only if there is an edge from $f_j$ to $f_i$. We will thus view the causality graph as an {\it undirected\/} graph. 

%Finally, in the interest of  brevity we sometimes refer to a potential causality graph as a causality graph.

%Moreover, in this work we will consider undirected causality graphs, and we will sometimes write $i \sim j$ to express the fact that $j \in \Gamma(i) $ (and, thus, $i \in \Gamma(j))$.  

%  For a state $\sigma$, let $e_{\sigma }$ denote  the indicator vector of $\sigma$, i.e., $e_{\sigma}[\sigma] = 1$ and  $e_{\sigma}[\tau] = 0 $ for all $\tau \in \Omega \setminus \{ \sigma \}$. For a fixed operator norm $\| \cdot \|$ and  invertible matrix $M$ such that $ \sum_{\sigma \in \Omega} \| M e_{\sigma}  \| = 1$, we define the \emph{charge} of flaw $f_i$ to be $\gamma_i = \| M A_i M^{-1}\| $. Finally, we denote by $\| \cdot \|_*$ the dual norm of $\|  \cdot \|$.
%
For a given probability measure $\mu$ supported on the state space $\Omega$, and for each flaw $f_i$, we define the \emph{charge} 
\begin{align}\label{charge_def}
\gamma_i =  \max_{\tau \in \Omega } \sum_{ \sigma \in f_i } \frac{ \mu(\sigma) }{ \mu(\tau) }  \rho_i(\sigma,\tau) .
\end{align}
In Section~\ref{connection} we give the intuition behind the definition of charges and also draw a connection with the parameters $b_i$ in Theorem~\ref{generalLLLL}. We are now ready to state the main result of~\cite{AIK}.

\begin{theorem}[\cite{AIK}]\label{our_theorem}
Assume that, at each step, the algorithm chooses to address the lowest indexed flaw according to an arbitrary, but fixed, permutation of $[m]$. If there exist positive real numbers $x_i\in(0,1)$ for
$1\le i\le m$ such that
\begin{align}\label{generalAlgoLLL}
\gamma_i \le (1-\epsilon) x_i \prod_{j \in \Gamma(i)} (1-x_j)  \enspace \text{ for every $i \in [m]$}
\end{align}
for some $\epsilon \in (0,1)$, then the algorithm reaches a flawless object within $(T_0+s)/\epsilon$ steps with probability at least $1-2^{-s}$, where
\[
T_0 	=  \log_2 \left(  \max_{  \sigma \in \Omega} \frac{ \theta(\sigma)}{ \mu(\sigma) }  \right)  +  \sum_{j \in [m] } \log_2 \left( \frac{1}{ 1-x_j} \right).
\]
\end{theorem}

We also describe another  theorem that can be used to show convergence in a polynomial number of steps, even when the number of flaws is super-polynomial, assuming that the algorithm has a nice 
``commutativity" property which we describe next.

\begin{definition}\label{commutative_algos}
For $i  \in [m]$, let $A_i $  denote the $|\Omega| \times |\Omega|$ matrix defined by $A_i[ \sigma, \sigma' ] = \rho_i(\sigma,\sigma')$ if $\sigma \in f_i$, and $A_i[\sigma,\sigma'] = 0$ otherwise. An algorithm defined by matrices $A_i$, $i \in [m]$, is \emph{commutative with respect to a causality relation $\sim$} if   for every $i,j \in [m]$ such that $i \nsim j$ we have  $A_i A_j = A_j A_i$.
\end{definition}

\begin{remark}\label{comm_remark}
 As shown in~\cite{Kolmofocs,LLLWTL,HIK}, commutative algorithms have several additional nice properties: they are often parallelizable, their output distribution approximates the so-called ``LLL-distribution", etc. Here we  use the fact that commutative algorithms converge in polynomial time even in the presence of superpolynomially many flaws, assuming that the causality graph can be covered by a polynomial number of cliques (see Theorem~\ref{clique_decomp} below). It is also worth noting that, if there were an efficient parallel algorithm for sampling matchings in multigraphs,  namely a parallel version of  the  MCMC algorithm of Theorem~\ref{sampling_matchings} which we discuss in the next section and which we use in our algorithm for Theorem~\ref{main}, then our proof directly implies a parallel algorithm for Theorem~\ref{main}.  The study of parallel versions of MCMC sampling algorithms has been initiated recently in~\cite{DBLP:journals/dc/FengSY20,DBLP:conf/wdag/FischerG18}.
\end{remark}

We note that Definition~\ref{commutative_algos} was introduced in~\cite{HIK}, as a generalization of the combinatorial definition of commutativity introduced in~\cite{Kolmofocs}. While the latter would suffice for our purposes, we choose to work with Definition~\ref{commutative_algos} due to its compactness.

\begin{theorem}\label{clique_decomp}
Let $\mathcal{A}$ be a commutative algorithm  with respect to  a causality relation $\sim$.  Assume there exist positive real numbers $\{x_i\}_{i \in [m] }$ in $(0,1)$ such that condition~\eqref{generalAlgoLLL} holds. Assume further that the causality graph induced by $\sim$ can be  covered by $n$ cliques  with potentially further edges between them. Setting  $\delta := \min_{ i \in [m] } x_i \prod_{j \in \Gamma(i) } (1 - x_j) $,  the expected number of steps performed by $\mathcal{A}$ is at most $t = O \left(  \max_{  \sigma \in \Omega} \frac{ \theta(\sigma)}{ \mu(\sigma) } \cdot \frac{n}{\epsilon}  \log \frac{ n \log ( 1/ \delta) }{ \epsilon} \right) $, and for any parameter $\lambda \ge 1$,  $\mathcal{A}$ terminates within $\lambda t$ resamplings with probability $1 - \mathrm{e}^{ -\lambda }$.
\end{theorem}

As shown in \cite[Theorem 3.2]{LLLWTL}, the proof of Theorem~\ref{clique_decomp} reduces to that
%\marginpar{\tiny AS: Reworded this: PLSE CHECK!}
of the analogous result  of Hauepler, Saha and Srinivasan~\cite{Haeupler_jacm} for the Moser-Tardos algorithm, and hence we omit it.
%Following Theorem 3.2 in~\cite{LLLWTL},  the proof of Theorem~\ref{clique_decomp} is  identical to the analogous result  of Hauepler, Saha and Srinivasan~\cite{Haeupler_jacm} for the Moser-Tardos algorithm and hence we omit it.

\subsection{Hard-Core Distributions on Matchings}

A probability distribution $\nu$ on the matchings of a  multigraph $G$ is \emph{hard-core} if it is
%\AS{Throughout this section, changed $\mu$ to $\nu$ for the probability distribution (for consistency
%with later use).  Plse check for collateral damage!!!   \red{OK } }
obtained by associating to each edge $e$ a positive real $\lambda(e)$ (called the \emph{activity} of $e$) so that the probability of any matching $M$ is proportional to $\prod_{e \in M} \lambda(e)$. Thus, recalling that $\mathcal{M}(G)$ denotes the set of matchings of $G$, and setting $\lambda(M) = \prod_{e \in M} \lambda(e)$ for each $M \in \mathcal{M}(G)$, we have
\begin{align*}
\nu(M) = \frac{ \lambda(M) }{ \sum_{M' \in \mathcal{M}(G) }  \lambda(M')}  .
\end{align*}

The characterization of the matching polytope due to Edmonds~\cite{edmonds1965maximum}  and a result of Lee~\cite{lee1990some} (which was also shown independently by Rabinovich et al.~\cite{rabinovich1992quadratic}) imply the following connection between fractional edge colorings and hard-core probability distributions on matchings. Before describing it, we need a definition.

For any probability distribution $\nu$ on the matchings of a multigraph $G$, we refer to the probability that a particular edge $e$ is in the random matching as the \emph{marginal}  of $\nu$ at~$e$.
We write $(\nu_{e_1},\ldots,\nu_{e_{|E(G)|}})$ for the collection of marginals of~$\nu$ at all the edges
$e_i\in E(G)$.
% Further, we say that $\nu$ has marginals $(\nu_{e_1}, \ldots, \nu_{e_{|E(G)|})})$ to express that the marginal of $\nu$ at edge $e_i$ is $\nu_{e_i} $.

\begin{theorem}[\cite{lee1990some,rabinovich1992quadratic}] \label{Alistair_hard_core}
There is a hard-core probability distribution $\nu$ with marginals $(\frac{1}{c}, \ldots, \frac{1}{c} )$ if and only if there is a fractional $c'$-edge coloring of $G$ with $c'<c$, i.e., if and only if $\chi_e^* < c$.
\end{theorem}
Kahn and Kayll~\cite{kahn1997stochastic}  proved that the probability distribution promised by Theorem~\ref{Alistair_hard_core} is endowed with very useful approximate stochastic independence properties.
\begin{definition}
Suppose  we choose a random matching $M$ from some probability distribution.  We say that an event $Q$ is \emph{$t$-distant} from a vertex~$v$ if $Q$ is completely  determined by the choice of all matching edges at distance at least~$t$ from~$v$. We say that $Q$ is $t$-distant from an edge $e$ if it is $t$-distant from both endpoints of~$e$.
\end{definition}
%Kahn  and Kayll essentially showed that, for any edge $e$, the probability of the event that it belongs to a matching drawn from the probability distribution promised by Theorem~\ref{Alistair_hard_core} is almost independent from any $t$-distant event (with respect to $e$), for large enough $t$. 
%\newpage
\begin{theorem}[\cite{kahn1997stochastic}] \label{hardcore_dist}
For any $\delta > 0$, there exists a $K = K(\delta)$ such that for any multigraph $G$ with fractional chromatic index $c$ there is a hard-core distribution $\nu$ with marginals $(\frac{1- \delta}{c }, \ldots, \frac{1-\delta}{c}   )$ such that: 
\begin{enumerate}[(a)]

\item for every $e \in E(G)$, $\lambda(e) \le \frac{K}{c} $ and hence $\forall v \in V(G)$, $\sum_{e \ni v } \lambda(e) \le K $;

\item 
for every $\epsilon \in (0,1)$, if we choose a matching $M$ according to $\nu$ then, 
%\marginpar{\tiny AS: In display, changed to $\in(1\pm\epsilon)$ notation as it looks better; not sure this is OK???}
for any edge $e$ and event $Q$ which is $t$-distant from $e$,
\begin{align*}
\Pr( e \in M \mid Q )       
\in (1\pm \epsilon) \Pr (e \in M ),
%(1- \epsilon ) \Pr[ e \in M ] &\le \Pr[  e \in M \mid Q ] \\         
%				     &\quad\qquad\le (1+ \epsilon) \Pr[ e \in M ],
\end{align*}
where $t= t(\epsilon) =  8 (K+1)^2 \epsilon^{-1} + 2$.
\end{enumerate}
\end{theorem}

We conclude this subsection by stating the result of Jerrum and Sinclair~\cite{Alistair1} for sampling from
hard-core distributions on matchings. We also describe  a few of its applications that will be helpful in our proofs.  The algorithm of~\cite{Alistair1} works by simulating a rapidly mixing Markov
chain on matchings, whose stationary distribution is the desired hard-core distribution~$\nu$,
and outputting the final state.

\begin{theorem}[\cite{Alistair1}, Corollary~4.3]\label{sampling_matchings}
Let $G$ be a multigraph,  $\{\lambda(e)\}_{e\in E(G)}$ a vector of activities associated
with the edges of~$G$, and $\nu$ the corresponding hard-core distribution.  Let $n=|V(G)|$ be the number of vertices of $G$ and define 
$\lambda' = \max \{   \max_{u,v \in V(G) } \sum_{ e \ni \{u, v \} } \lambda(e) ,1  \} $.
There exists an algorithm that, for any
$\epsilon>0$, runs in time $\poly(n,\lambda',\log\epsilon^{-1})$ and outputs a matching
in~$G$ drawn from a distribution~$\nu'$ such that $\|  \nu - \nu'  \|_{\mathrm{ TV} } \le \epsilon$.
\end{theorem}

%\begin{theorem}\label{sampling_matchings}
%Let $G$ be a multigraph,  $\lambda$ be the vector of activities   associated with the edges of $G$, and $\nu$ be the corresponding hard-core distribution. For every $\epsilon > 0$ there exists an ergodic Markov chain with stationary distribution $\nu$ and mixing time 
%\begin{align*}
%\tau(\epsilon ) \le 4 |E(G)| \lambda' \left( n  \left( \ln n + \ln \lambda'\right) +  \ln \epsilon^{-1} \right)  ,
%\end{align*} 
%where $\lambda' = \max \{   \max_{v,u \in V(G) } \sum_{ e \ni \{u, v \} } \lambda(e) ,1  \} $ and $n = |V(G)|$.
%\end{theorem}

\begin{remark}
\cite{Alistair1} establishes this result for matchings in (simple) graphs.  However, the extension
to multigraphs is immediate: make the graph simple by replacing each set of multiple edges $e_1,\ldots,e_\ell$ between a pair of vertices $u,v$ by a single edge~$e$ of activity~$\lambda(e)=\sum_i \lambda(e_i)$; then use the algorithm to sample a matching from the hard-core distribution
in the resulting simple graph;
finally, for each edge $e=\{u,v\}$ in this matching, select one of the corresponding multiple
edges $e_i\ni\{u,v\}$ with probability $\lambda(e_i)/\sum_i\lambda(e_i)$.  Note that the
running time will depend polynomially on the maximum activity~$\lambda'$ in the simple graph,
as claimed.
\end{remark}

\blue{Note that, via a standard argument, the algorithm of  Theorem~\ref{sampling_matchings} can be used to design a \emph{fully-polynomial randomized approximation scheme (f.p.r.a.s.)}  for the \emph{partition function} of a hard-core probability distribution on the matchings of a multigraph $G$ — namely,  for the quantity $Z_{\lambda}(G) =  \sum_{M \in \mathcal{M}(G) } \lambda(M) $.

\begin{theorem}[\cite{Alistair1}, Corollary~4.4]\label{fpras_matchings}
Let $G$ be a multigraph,  $\{\lambda(e)\}_{e\in E(G)}$ a vector of activities associated
with the edges of~$G$, and $Z_{\lambda}(G)$ the corresponding partition function.  Let $n=|V(G)|$ be the number of vertices of $G$ and define 
$\lambda' = \max \{   \max_{u,v \in V(G) } \sum_{ e \ni \{u, v \} } \lambda(e) ,1  \} $. There exists an algorithm that, for any $\epsilon > 0$, runs in time $\mathrm{poly}(n,\lambda', 1/\epsilon)$ and outputs a quantity $\widetilde{Z}_G(\lambda)$ such that
$ \Pr \left(  (1-\epsilon)Z_{G}(\lambda)     \le \widetilde{Z}_G(\lambda)     \le (1+\epsilon) Z_{G}(\lambda)   \right) \ge 3/4.$ 
\end{theorem}

\begin{remark}\label{repetition}
The estimate in Theorem~\ref{fpras_matchings} could be arbitrarily bad with probability $1/4$. However, this probability can be
reduced to any desired  $\delta > 0$ by performing $O( \log \delta^{-1}) $ trials and taking the median.
\end{remark}
Theorem~\ref{fpras_matchings} allows us to design  a f.p.r.a.s. for the edge-marginals of a hard-core probability distribution  on the matchings of a multigraph $G$.
\begin{corollary}\label{easy_cor}
Let $G$ be a multigraph,  $\{\lambda(e)\}_{e\in E(G)}$ a vector of activities associated
with the edges of~$G$, and $\nu$ the corresponding hard-core distribution.  Let $n=|V(G)|$ be the number of vertices of $G$ and define 
$\lambda' = \max \{   \max_{u,v \in V(G) } \sum_{ e \ni \{u, v \} } \lambda(e) ,1  \} $. There exists an algorithm that, for any edge $e$, $\epsilon > 0$ and $\delta > 0$,  runs in time $\mathrm{poly}(n,\lambda', 1/\epsilon, \log \delta^{-1})$ and outputs a quantity $\widetilde{\nu}_e$ such that
$ \Pr \left(  (1-\epsilon) \nu_e \le  \widetilde{\nu}_e  \le ( 1+\epsilon)  \nu_e  \right) \ge 1- \delta$, where $\nu_e$ is the marginal of $\nu$ at $e$.
\end{corollary}
\begin{proof}
Let $G_e$ be the mutligraph  obtained by removing $e$ along with every other edge of $G$ adjacent to it. Let $Z_{\lambda}(G)$, $Z_{\lambda}(G_e)$ denote the partition functions corresponding to multigraphs $G, G_e$ with respect to $\{\lambda(e)\}_{e\in E(G)}$, respectively. Observe now that $\nu_e = \lambda(e) \cdot Z_{\lambda}(G_e)/ Z_{\lambda}(G)  $. Using the f.p.r.a.s. promised by Theorem~\ref{fpras_matchings} (and Remark~\ref{repetition}) to get appropriately accurate estimates for $Z_{\lambda}(G), Z_{\lambda}(G_e) $, we directly obtain an estimate for $\nu_e$ that satisfies the  guarantees of Corollary~\ref{easy_cor}.

\end{proof}

Finally, one can use Theorem~\ref{fpras_matchings}  as a subroutine in the algorithm of Singh and Vishnoi~\cite{singh2014entropy} to obtain the following result. 

\begin{corollary}\label{computing_the_dist}
Let $G$ be a multigraph on $n$ vertices and let $\delta \in (0,1)$ be a parameter. Let $\nu = \nu_{\delta}$ be the hard-core probability distribution over the matchings of $G$ promised by Theorem~\ref{hardcore_dist}. 
For every $\eta > 0$, there exists a $\mathrm{poly }(n, \log \eta^{-1}, \log \delta^{-1} )$-time algorithm that computes a set of edge activities $ \{ \lambda'(e) \}_{e \in E(G)} $ such that  the corresponding hard-core distribution $\nu'$ satisfies $\|  \nu - \nu'  \|_{\mathrm{ TV} } \le \eta $.
\end{corollary}
\begin{proof}
Corollary~\ref{computing_the_dist} follows in a straightforward way from the main results of Singh and Vishnoi~\cite{singh2014entropy} and  Jerrum and Sinclair~\cite{Alistair1}. Briefly, the main result of~\cite{singh2014entropy} states that finding a distribution that approximates $\nu$ can be seen as the solution of a max-entropy distribution estimation problem which can be efficiently solved given a ``generalized  counting oracle" for~$\nu$. The latter oracle is provided by Theorem~\ref{fpras_matchings}.
\end{proof}

}

\section{Causality, Lopsidependency and Approximate Resampling Oracles}\label{connection}

In this section we show  a connection between Theorem~\ref{generalLLLL} and  Theorem~\ref{our_theorem}.  While this section is not essential to the proof of our main results, it does provide useful intuition  since it implies the following natural approach to making applications of the Lopsided LLL algorithmic: we start designing a local search algorithm for addressing the flaws that correspond to bad events by considering the family of probability distributions $\{ \rho_i(\sigma,\cdot) \}_{i \in [m],  \sigma \in f_i}$ whose supports induce a causality graph that coincides with the lopsidependency graph of the Lopsided LLL application of interest. This is typically  a straightforward task. The key to successful  implementation is our ability to make the way in which the algorithm addresses flaws sufficiently  \emph{compatible} with the underlying probability measure $\mu$. To make this precise, we first recall an algorithmic interpretation of the notion of {\it charges\/} defined in~\eqref{charge_def}.

As shown in~\cite{AIK}, the charge $\gamma_i$ captures the compatibility between the actions of the algorithm for addressing flaw $f_i$ and the measure~$\mu$.  To see this,  consider the probability, $\nu_i(\tau)$, of ending up in state $\tau$ after (i) sampling a state $\sigma \in f_i$ according to $\mu$, and then (ii) addressing $f_i$ at $\sigma$. Define the \emph{distortion} associated with $f_i$ as
\begin{align}
d_i :=  \max_{\tau \in \Omega } \frac{\nu_i (\tau) }{\mu (\tau) }  \ge 1,
\end{align}
i.e., the maximum possible inflation of a state probability incurred by addressing $f_i$ (relative to its probability under $\mu$, and averaged over the initiating state $\sigma \in f_i$ according to $\mu$). Now observe from~\eqref{charge_def} that
\begin{align}\label{eq:charges}
\gamma_i =  \max_{\tau \in \Omega}  \frac{1}{ \mu(\tau) } \sum_{ \sigma \in f_i } \mu( \sigma )  \rho_i(\sigma, \tau ) =  d_i \cdot  \mu(f_i). 
\end{align}
An algorithm for which $d_i = 1$ is called a \emph{resampling oracle}~\cite{HV} for $f_i$,  and notice that it perfectly removes the conditional of the addressed flaw. However, designing resampling oracles for sophisticated measures can be impossible  by local search. This is because small, but non-vanishing, correlations can travel arbitrarily far in $\Omega$. Thus, allowing for some distortion can be very helpful, especially in cases where correlations decay with  distance.

%\newpage

Theorem~\ref{causality_lopsi} below shows that Theorem~\ref{our_theorem} is the  algorithmic counterpart of Theorem~\ref{generalLLLL}.
\begin{theorem}\label{causality_lopsi}
Given a family of flaws $F = \{f_1, \ldots, f_m \}$ over a state space $\Omega$, an algorithm $\mathcal{A}$ with causality graph $C$ with neighborhoods $\Gamma( \cdot)$, and  a measure $\mu$ over $\Omega$, then for each $S \subseteq F \setminus \Gamma(i)$ we have 
\begin{align}\label{connection_rel}
 \mu\Bigl(f_i \mid \bigcap_{j \in S} \overline{f_j}\Bigr)  \le \gamma_i ,
\end{align}
where the $\gamma_i$ are the charges of the algorithm as defined in~\eqref{charge_def}.
\end{theorem}

\begin{proof}

Let $F_S :=  \bigcap_{j \in S}  \overline{f_j} $.
Observe that
\begin{eqnarray}
 \mu(f_i \mid F_S)                 & = & \frac{ \mu( f_i \cap  F_S  )   }{ \mu( F_S  ) }  \nonumber \\
 								         & = & \frac{ \sum_{\sigma \in f_i \cap F_S}  \mu(\sigma) \sum_{ \tau \in a(i,\sigma)}  \rho_i(\sigma,\tau)  }{ \mu(  F_S )   }   \nonumber \\
 									 & = &  \frac{ \sum_{\sigma \in f_i \cap  F_S }  \mu(\sigma) \sum_{ \tau \in  F_S}  \rho_i(\sigma,\tau)  }{ \mu(  F_S)   }  \label{eqn:star},
\end{eqnarray}									 
where the second equality holds because each $\rho_i(\sigma,\cdot)$ is a probability distribution, and  the third by the definition of causality and the fact that $S \subseteq F \setminus \Gamma(i)$.	Now notice that changing the order of summation in~\eqref{eqn:star} gives			 
\begin{align*}						 
&\frac{ \sum_{ \tau \in F_S} \sum_{\sigma \in f_i \cap  F_S}  \mu(\sigma)  \rho_i(\sigma,\tau)  }{ \mu(  F_S)   }  \\
  & =   \frac{ \sum_{ \tau \in  F_S}  \mu(\tau) \sum_{\sigma \in f_i \cap  F_S}   \frac{ \mu(\sigma) }{ \mu(\tau) }  \rho_i(\sigma,\tau)  }{ \mu(  F_S)   }     \\
  &\le  \frac{ \sum_{ \tau \in  F_S} \mu(\tau) \left( \max_{\tau' \in \Omega} \sum_{\sigma \in f_i  }   \frac{\mu(\sigma)}{\mu(\tau') }   \rho_i(\sigma,\tau') \right) }{ \mu( F_S)   }   \\
  & = \gamma_i .
\end{align*}
\end{proof}

In words, Theorem~\ref{causality_lopsi} shows that causality graph $C$ is a lopsidependency graph  with respect to measure $\mu$ with $b_i = \gamma_i$ for all $i \in [m]$.  Thus, when designing  an algorithm for an application of Theorem~\ref{generalLLLL} using Theorem~\ref{causality_lopsi}, we have to make sure that the induced causality graph coincides with the lopsidependency  graph, and that the measure distortion induced when addressing flaw $f_i$ is sufficiently small so that the resulting charge $\gamma_i$ is bounded above by $b_i$.

\section{Edge Coloring Multigraphs: Proof of Theorem~\ref{main}}\label{main_1}

We follow the exposition of the proof of Kahn in~\cite{mike_book}. \blue{Note that throughout the proof we assume that the maximum degree $\Delta$ of  the input multigraph $G$ satisfies $\Delta \ge \Delta_0$ for some appropriately large constant $\Delta_0$.}

 The key to the proof of Theorem~\ref{main} is the following lemma.
%\newpage
\begin{lemma}\label{main_lemma}
For all $\epsilon >0$, there exists $\chi_0 = \chi_0(\epsilon)$  such that if $ \chi_e^*(G) \ge \chi_0$ then we can find $N = \lfloor \chi_e^*(G)^{ \frac{3}{4} }  \rfloor $ matchings in $G$ whose deletion leaves a multigraph $G'$ with $\chi_e^*(G') \le \chi_e^*(G) - (1+\epsilon)^{-1} N $ in expected $\mathrm{poly}(n, \ln \frac{1}{ \epsilon} )$ time. 
%\AS{Added "for any constant $c>0$.  OK? \red{OK}}
\end{lemma}

\begin{remark}
Since $\chi_e^*(G) = \mathrm{poly}(n)$, we may assume that $\epsilon  \ge \frac{1}{ \mathrm{poly}(n) }$ without loss of generality . Therefore, the expected running time of the algorithm promised by Lemma~\ref{main_lemma}  is $\mathrm{poly}(n)$.
\end{remark}

Using the algorithm of  Lemma~\ref{main_lemma} recursively, for every $\epsilon >0 $ we can efficiently find an edge coloring of $G$ using at most   $ (1+\epsilon) \chi_e^* + \chi_0$ colors as follows.  First, we compute $\chi_e^*(G)$ using the algorithm of Padberg and Rao~\cite{padberg1982odd}. If $\chi_e^* \ge  \chi_0$, then we apply Lemma~\ref{main_lemma} to get a multigraph $G'$ with $\chi_e^*(G') \le \chi_e^*(G) - (1+\epsilon)^{-1} N $. We can now color $G'$ recursively using at most $(1+\epsilon) \chi_e^*(G') + \chi_0 \le (1+\epsilon) \chi_e^*(G)  - N +\chi_0 $ colors.  Using one extra color for  each of the $N$ matchings promised by Lemma~\ref{main_lemma}, we can then complete the coloring of $G$, proving the claim.  In the base case where $\chi_e^*(G) < \chi_0$,    we color $G$ greedily using  $2\Delta-1 $ colors. The fact that
%\AS{Changed $2\Delta + 1$ to $2\Delta - 1$.  Plse check the entire base case carefully!!! \red{OK}}
$ 2 \Delta -1  \le 2 \chi_e^* - 1 <  \chi_e^* + \chi_0$  concludes the proof  of Theorem~\ref{main} as the number of recursive calls is at most $n$.

In the following sections, we prove  Lemma~\ref{main_lemma}. In Section~\ref{chrom_algo} we describe the local search algorithm behind Lemma~\ref{main_lemma}, and in Section~\ref{natooos} we prove its convergence. In Sections~\ref{charge_lemma},~\ref{commutativity_lemma_proof}  we prove two important auxiliary lemmas that are used in our 
convergence proof.

\subsection{The Algorithm}\label{chrom_algo} 

Observe that we  only need to  prove Lemma~\ref{main_lemma}  for $\epsilon < \frac{1}{10}$ since, clearly, if it holds for $\epsilon$ then it holds for all $\epsilon' > \epsilon$.
So we fix $ \epsilon \in(0,0.1)$ and 
let $c^* = \chi_e^*(G) - (1+\epsilon)^{-1} N $. Our goal will be to delete $N$ matchings from $G$ to get a multigraph $G'$ which has fractional chromatic index at most $c^*$. 

\paragraph{The flaws.}
Let $\Omega = \mathcal{M}(G)^N $ be the set of  possible $N$-tuples of matchings of $G$. For a state $\sigma = (M_1, \ldots, M_N) \in \Omega$   let $G_{\sigma} $ denote the multigraph  obtained  by deleting the $N$ matchings $M_1, \ldots, M_N$ from $G$.  For a vertex $v \in V(G)$ we define $d_{G_{\sigma}}(v)$ to be the degree of $v$ in $G_{\sigma}$. We now define  the following flaws.  For every vertex $v \in V(G)$  let
\begin{align*}
    f_v   =    \left\{  \sigma \in \Omega: d_{G_{\sigma}}(v) >  c^*-   \frac{\epsilon}{4}  N \right\} .
\end{align*}
For every  connected  subgraph $H$ of $G$ with an odd number of vertices and such that (i)         $|V(H)| \le  \frac{8\Delta}{ \epsilon N } $, and (ii)  $|E(H)| >  \left( \frac{ |V(H)| -1 }{2 } \right) c^*$, let
\begin{align*}
  f_{H}      = \{ \sigma \in \Omega: H \subseteq G_{\sigma} \}  .
\end{align*}
The following lemma implies that it suffices to find a flawless state.  (This lemma was proved in~\cite{kahnChrom}, but we 
include a proof here for completeness.)

\begin{lemma}[\cite{kahnChrom}]\label{sufficient_flaws}
Any  flawless state $\sigma$ satisfies $\chi_e^*(G_{\sigma}) \le c^* $.
\end{lemma}
\begin{proof}
Edmonds' characterization~\cite{edmonds1965maximum} of the matching polytope implies that the chromatic index of $G_{\sigma}$ is at most $c^*$ if 
\begin{enumerate}
\item $\forall v:  d_{G_{\sigma} }(v) \le c^*$; and
\item  $\forall H \subseteq G_{\sigma} $ with an odd number of vertices: 
\begin{align*}
E(H) \le \frac{ |V(H)| -1 }{ 2}\, c^*.
\end{align*}
\end{enumerate}
Now clearly, addressing every flaw of the form $f_v$ establishes condition $1$. By summing degrees it also implies that, for every subgraph~$F$,
$|E(F)| \le \frac{ |V(F)|  (c^* -  \epsilon N/4)  }{2 }   \le \  \frac{| V(F)| }{2 }   c^* $.

Moreover, any odd subgraph $H$ can be decomposed into a connected component $H'$ with an odd number of vertices, and a 
(possibly empty) subgraph $F$ with an even number of vertices. 
%Thus, in the absence of $f_v$ flaws, it suffices to prove condition~$2$ for connected $H$. 
%Concretely, given a subgraph $H$ with an odd number of vertices, let $H'$ denote any of its connected components which has an odd number of vertices. Note that there should be at least one such component, otherwise $H$ would have an even number of vertices, which is a contradiction.  Let also $F$ be the graph induced from $H$ after removing $H'$. 
Since there are no edges between $F$ and $H'$, in the absence of $f_v$ flaws we obtain
\begin{align*}
|E(H) |= |E(F) |  + |E(H')|  \le \frac{ | V(F)|  c^*}{2}  +|E(H') |.
\end{align*}
Thus it suffices to prove condition~2 for the connected odd subgraph~$H'$, for if  $|E(H') | \le ( |V(H')| -1 ) c^*/2 $  then we have
\begin{align*}
|E(H) | \le ( |V(F)| + |V(H')| -1 ) c^*/2  =  ( |V(H) |  -1 ) c^* /2.
\end{align*}
Now, again by summing degrees, we see that if no $f_v$ flaw is present then condition~$2$ can fail only for $H$ with fewer than $\frac{8\Delta}{ \epsilon N }$ vertices, concluding the proof. Indeed, in the absence of $f_v$ flaws, we have $|E(H)| \le |V(H) | (c^* - \epsilon N/4)/2 $ and, since $c^* \le \chi_e^*(G) \le  2 \Delta$, 
if  $|V(H) | (c^* - \epsilon N/4)/2  \ge  (|V(H)| -1)c^* /2 $ then $|V(H)| \le c^*/( ( \epsilon/4)N )    \le  8\Delta/ \epsilon N   $.
\end{proof}

To describe an efficient algorithm for finding flawless states we need to (i) determine the initial distribution of the algorithm and show that it is efficiently samplable; (ii) show how to address each flaw efficiently; (iii) show that the expected number of steps of the algorithm is polynomial;  and finally (iv) show that we can search for flaws in polynomial time, so that each step is efficiently implementable.

\paragraph{The initial distribution.} 
Apply Theorem~\ref{hardcore_dist} with $\delta = \frac{ \epsilon }{ 4} $. Let  $\nu$ be the promised hard-core probability distribution, $\lambda=\{\lambda(e)\}$ the vector of activities associated with it, and $K$ the corresponding constant.  Note that the activities~$\lambda(e)$ defining~$\nu$
%\AS{Added some clarification here: plse check! \red{OK}}
are not readily available.
However, recalling Corollary~\ref{computing_the_dist} we see that we can efficiently compute a set of activities that gives an
arbitrarily good approximation to the desired distribution~$\nu$.

For a parameter $\eta >0$ and a distribution $p$, we  say that we \emph{$\eta$-approximately sample} from $p$ to  express the fact that we sample from a distribution $\tilde{p}$ such that $\| p - \tilde{p} \|_{\mathrm{TV} } \le \eta$. Set $\eta = \frac{1}{n^\beta}$, 
%\marginpar{\tiny AS: Changed $C$ to $\beta$ in three places in section~4!!!}
where $\beta$ is a sufficiently large constant to
be specified later,  and  let $\nu'$ be the distribution promised by Corollary~\ref{computing_the_dist}.  The initial distribution of our algorithm, $\theta$, is obtained by $\eta$-approximately sampling $N$ random matchings  (independently)  from~$\nu'$.  Observe that $\| \theta - \mu \|_{ \mathrm{TV}  }  \le 2 \eta N$,
where $\mu$ denotes the probability distribution over $\Omega$ induced by taking $N$ independent samples from $\nu$.

\paragraph{Addressing flaws.} 
For an integer $d> 0 $ and a connected subgraph $H$,  let $S_{ <d }(H) $ be the set of vertices within distance strictly less than~$d$ of a vertex $u \in V(H) $.    
Given a state $\sigma = (M_1, \ldots, M_N)  $, a subgraph $H$,  and  $d >0$  let 
 \begin{align*}
 Q_H(d,\sigma) = \left(M_1 - S_{<d}(H),  \ldots, M_N - S_{ <d}(H)  \right) ,
 \end{align*}
where we define $M- X = M \cap E(G-X) $. Moreover,   let $Q_H^i(d,\sigma) = M_i - S_{< d }(H)  $ denote the $i$-th entry of $Q_H(d,\sigma)$.  \blue{(In words, $Q_H^i(\sigma,d)$ is  the set of edges of $M_i$ with the property that both their endpoints are at distance at least $d$ from $H$.)} 
Finally, let $G_{<d+1}(H)$  be the multigraph induced by $S_{<d+1}(H)$ and   $ \mathcal{M}_{d+1}^i = \mathcal{M}_{d+1}^i(H,\sigma) $  be the set of matchings of $G_{<d+1}(H) $ that are ``compatible" with $Q_H^i(d,\sigma)$. That is, for any  matching $M$ in $ \mathcal{M}_{d+1}^i$ we have that $M \cup Q_H^i(d,\sigma) $ is also a matching of $G$. More specifically, note that $\mathcal{M}^i_{d+1}(H,\sigma) $ corresponds to  the set of matchings of the following multigraph $G_{i,<d+1}(H)$. Let $V_{i,d}$ denote the set of vertices of $S_{ <d+1}(H)$ that belong to edges in $Q_H^i(\sigma,d ) $. Multigraph $G_{i,<d+1}(H)$ is  induced by $S_{ <d+1}(H) \setminus V_{i,d}$.

 We  consider the  procedure {\sc Resample}  below which takes as input a connected subgraph $H$,
%\AS{After the recent change in this procedure, should we be surprised that it had no effect on the
%rest of the analysis? \red{Yes. See my reply to the first margin in Section 4.3 for details. } \\ \green{See email.}}
%\AS{In line 4 below, should be "multigraph induced..." right?}
a state $\sigma$ and a positive integer $d \le n$, and which will be used to address flaws.

\begin{algorithm}
\begin{algorithmic}[1] 
\Procedure{Resample}{$H,\sigma,d$}
\State Let $\sigma = (M_1, M_2, \ldots, M_N)$
\For {$i=1$ to $N$}
	%\State Let $E_{i, \ge d}$ be the set of  edges of $M_i$ that do not belong to the multigraph induced by $S_{<d+1}(H)$
	%\State Let $E_{i,= d}$ be the set of edges  of $M_i$ both of whose endpoints are at distance exactly $d$ from $H$
	%\State Let $V_{i,d}$ be the set of vertices of $S_{ <d+1}(H)$ that belong to edges in $E_{i,\ge d} \cup E_{i, = d }$
	%\State  Let $G_{i,<d+1}$ be the multigraph  induced by $S_{ <d+1}(H) \setminus V_{i,d}$  \label{matchings_def} 
	\State Let $p$ be the hard-core distribution over matchings in $\mathcal{ M}^i_{d+1}$  induced by  $ \{ \lambda'(e) \}_{ e \in E(G_{<d+1} ) }  $. \label{seven}
	\State $\eta$-approximately sample a matching $M$ from  distribution $p$ \label{eight}
	\State Let $M_i' =  M \cup Q_H^i(d,\sigma) $ \Comment By definition, $M_i'$ is a matching
\EndFor
\State Output $\sigma' = (M_1', M_2', \ldots, M_N')$
\EndProcedure
\end{algorithmic}
\end{algorithm}

 Throughout the proof, we fix the parameter 
\begin{align*}
t = 8 (K+1)^2 \delta^{-1} + 2.
\end{align*}
To address $f_v, f_H$ in state $\sigma$, we  invoke procedures {  \sc Resample $(  \{ v \}  ,\sigma, t)$ } and {\sc Resample $(  H ,\sigma, t)$}, respectively.

%\begin{figure}
%\hspace*{2cm}\includegraphics[width = 0.7  \textwidth]{Example}\hspace*{-2cm}
%\end{figure}

\paragraph{Searching for flaws.} 
Notice that we can compute $c^*$ in polynomial time using the algorithm of Padberg and Rao~\cite{padberg1982odd}. Therefore,  given a state $\sigma \in \Omega$, we can search for flaws of the form $f_v$ in polynomial time. However, the flaws of the form $f_H$ are potentially exponentially many, so a brute-force search does not suffice for our purposes. 

Fortunately, the result of Padberg and Rao provides  a polynomial time oracle for this problem as well. Recall Edmonds' characterization used in the proof of Lemma~\ref{sufficient_flaws}. The constraints over odd subgraphs~$H$ are called \emph{matching constraints}. Recall further that in the proof of Lemma~\ref{sufficient_flaws} we showed that, in the absence of $f_v$-flaws, the only matching constraints that could possibly be violated correspond to $f_H$ flaws. On the other hand, the
oracle of Padberg and Rao can decide in polynomial time  whether $G$ has a fractional $c$-coloring or return a  violated matching constraint, for every $c \ge 0$. Hence, if our algorithm prioritizes $f_v$ flaws over $f_H$ flaws, this oracle can be used to detect the latter in polynomial time.

\subsection{Proof of Lemma~\ref{main_lemma}}\label{natooos}

We are left to show that the expected number of steps of the algorithm is polynomial and that each step can be executed in polynomial time. To that end, we will show that both of these statements are true assuming that the initial distribution $\theta$ is $\mu$ instead of approximately $\mu$, and that in Lines~\ref{seven},~\ref{eight} of the procedure {\sc Resample$(H,\sigma,d)$}  we perfectly sample from the hard-core probability distribution induced by   activities $\{ \lambda(e) \}_{ e \in E(G_{i,<d}(H) ) }  $ instead of $\eta$-approximately sampling from $p$. We can maximally couple the approximate and ideal distributions, and then take the constant~$\beta$ in the definition of the approximation parameter $\eta$ to be sufficiently large. The latter  implies that the probability that the coupling will fail during the execution of the algorithm is negligible (i.e., at most $ \frac{1}{n^c}$). Since the fractional chromatic index of a multigraph can be computed in polynomial time, we can absorb the probability that the coupling fails into the polynomial expected running time by executing our algorithm sufficiently many times. That is, we execute our algorithm for a number of steps that is twice its expected running time, and if the edge coloring it produces is not a desirable one, we repeat the process. 

%\begin{remark}\label{coupling_remark} 
%The coupling argument we used in the previous paragraph, and which we use again later in the proof of Theorem~\ref{main_list}, can be stated a bit more formally  as follows. Let $F_{\mathrm{ideal} },  F_{\mathrm{approx} } $ be the indicator random variables of the events that our algorithm fails when we sample from the ideal and approximate distributions, respectively. Let $\mathcal{Q}$ be any \emph{coupling} of these two variables, i.e., a joint probability distribution that respects the marginals. We say that the coupling \emph{succeeds}  when we sample from $\mathcal{Q}$ and $F_{\mathrm{ideal} } = F_{\mathrm{approx} }$, and we say that the coupling \emph{fails} otherwise. Now we make use of the following elementary fact:
%\begin{align*}
%\Pr( F_{\mathrm{approx} } = 1 )   \le  \Pr(    F_{\mathrm{ideal }} = 1     )  + \Pr( \text{$\mathcal{Q}$ fails}  ). 
%\end{align*}
%Thus, we see that if the coupling succeeds with sufficiently high probability, then the assumption that our algorithm operates using the ideal distributions is correct up to a small additive error term in the probability of failure of the coupling. The latter can be subsumed in the running time of the algorithm by executing it  sufficiently many times.
%\end{remark}

For an integer $d > 0$ and a vertex $v$, let $F_{d}(v)$ be the set of flaws indexed by a vertex of $S_{< d }  (v)$ or a subgraph $H$ intersecting $S_{< d}(v)$. For each set $H$ for which we have defined $f_H$ we let $F_{d}(H) = \bigcup_{v \in V(H) } F_{d}(v)$. For each flaw $f_v$ we define the causality neighborhood $\Gamma(f_v) =  F_{t+2}(v)$, and for each flaw $f_H$ we define $\Gamma(f_H) = F_{t+2}(H)$, where $t$ is as defined in the previous subsection. Notice that this is a valid choice because flaw $f_v$ can only cause flaws in $F_{t+1}(v)$ and flaw $f_H$ can only cause flaws in $F_{t+1}(H)$. The reason why we choose these neighborhoods to be larger than seemingly necessary is because, as we will see, with respect to this causality graph our algorithm is commutative, allowing us to apply Theorem~\ref{clique_decomp}.
\begin{lemma}\label{termination}
Let $f \in \{f_v, f_H \}$  for a vertex $v$ and a  connected subgraph $H$ of $G$ with an odd number of vertices and let  $D =  \Delta^{ t+ 2\Delta^{ \frac{1}{3}} +4 }  $.  We have:
\begin{enumerate}[(a)]

\item       $\gamma_f \le \frac{1}{2\mathrm{e} D }  $  ;

\item$ |\Gamma(f)| \le D$,
\end{enumerate}
where the charges are computed with respect to the measure $\mu$ and the algorithm that samples from the ideal distributions.
\end{lemma}

 \begin{lemma}\label{commutativity_lemma}
For each pair of flaws $f \nsim g$, the matrices $A_f, A_g$ commute.
\end{lemma}
The proof of Lemma~\ref{termination} can be found in Section~\ref{charge_lemma}. Lemma~\ref{commutativity_lemma} establishes that our algorithm is commutative with respect to the causality relation $\sim$ induced by neighborhoods $\Gamma(\cdot)$. Its proof can be found in Section~\ref{commutativity_lemma_proof}.

Now, setting $x_f = \frac{1}{ 1+ \max_{ f' \in F }| \Gamma(f') | }$ for each flaw $f$, we see that 
%\AS{Please double-check this calculation once more! \red{OK}}
condition~\eqref{generalAlgoLLL}  with      $ \epsilon = 1/4$  is implied by
\begin{align}\label{symmetric}
 \gamma_f  \cdot \bigl(1+\max_{f' \in F} ( | \Gamma(f') |\bigr)  \cdot  \mathrm{e} \le    3/4   \enspace \text{for every flaw $f$} ,
\end{align}
which is true for large enough $\Delta$ according to Lemma~\ref{termination}.  Notice further that the causality graph induced by $\sim$ can be  covered by $n$ cliques, one for each vertex of  $G$,  with potentially further edges between them.  Indeed, flaws indexed by subgraphs  that contain  a certain vertex  of $G$ form a clique in the causality graph. Combining  Lemma~\ref{commutativity_lemma}
%\AS{Note that running time has dropped from $\tilde O(|E|^{4/3})$ to $\tilde O(n)$.  Plse double-check!!! \red{OK}}
with the latter observation, we are able to apply  Theorem~\ref{clique_decomp}, which implies that our algorithm terminates after an expected number of at most $O\bigl( \max_{\sigma \in \Omega} \frac{ \theta(\sigma) }{\mu(\sigma) } \cdot n \log ( n \log (1/\delta) ) \bigr) = O( n  \log n   ) $ steps. (This is because we assume that $\theta = \mu$ per our discussion above.)

This completes the proof of Lemma~\ref{main_lemma} and hence, as explained at the beginning of
%\AS{Added this final paragraph \red{OK}}
Section~\ref{main_1}, Theorem~\ref{main} follows.  It remains, however, to go back and
prove Lemmas~\ref{termination} and~\ref{commutativity_lemma}, which we do in the next two subsections.

\subsection{Proof of Lemma~\ref{termination}}\label{charge_lemma}

%In this section we prove Lemma~\ref{termination}.
%\AS{Seems that this stuff should look obviously equivalent to the defns. in the Resample algorithm? \red{Yes you are right. Here's the ``issue" though.   $\mathcal{M}_{t+1}(v, \sigma) $ is the set of matchings we would like to sample from efficiently. To do so, we need to show that $\mathcal{M}_{t+1}(v,\sigma)$ corresponds to a set of matchings of a specific graph  and then apply your MCMC algorithm. This is what is essentially established in the description of { \sc Resample}.    \\ How about (i) we define $\mathcal{M}_{t+1}(v,\sigma) $  in Section 4.1; (ii) we use it in the definition of the Algorithm;  (iii) we explain in a Lemma that this can be done efficiently ?  \\ Another less ``intrusive" option is to add a small paragraph right after the definition of $\mathcal{M}_{d+1}(H,  \sigma)$ in this section that explains what's going on. }   \\ \green{See email.}  } 

 \begin{proof}[Proof of part (a)]
  
We will need the following key lemma, which was essentially  proved in~\cite{kahnChrom}. Its proof can be found in Appendix~\ref{bounds}. 
Recall that $\mu$ is the distribution over $\Omega$ induced by taking $N$ independent samples from $\nu$.
\begin{lemma}\label{basic_lemma}
 For any random state $\sigma $ distributed according to $\mu$:
\begin{enumerate}[(i)]
\item for every flaw~$f_v$ and state $ \tau \in \Omega $:  $\mu( \sigma \in f_v \mid Q_v(t,\sigma) = Q_v(t,\tau)  ) \le \frac{1 }{ 2 \mathrm{e} D}$; and
\item for every flaw $f_H$ and state $\tau \in \Omega$:  $\mu(  \sigma \in f_H \mid Q_H(t,\sigma) = Q_H(t,\tau)) \le  \frac{1}{ 2\mathrm{e} D } $.
\end{enumerate}
\end{lemma}

We show the proof of part (a) of Lemma~\ref{termination} only for the case of $f_v$- flaws, as the proof for $f_H$- flaws is very similar.   Specifically, our goal will be to prove that
\begin{align}\label{goal}
\gamma_{f_v}  =  \max_{\tau \in \Omega} \mu( \sigma \in f_v \mid Q_v(t,\sigma) = Q_v(t,\tau)  ) .
\end{align}
Lemma~\ref{basic_lemma} then concludes the proof.

Recalling the definition of~$\gamma_f$ from~\eqref{charge_def} we see that, in order to prove~\eqref{goal}, it suffices to show that, for $\sigma$ distributed according to  $\mu$ and any state $\tau \in \Omega$,
\begin{align}\label{third_charge}
\sum_{ \omega \in f_v } \frac{ \mu(\omega) }{ \mu(\tau) } \rho_{f_v} (\omega,\tau) =    \mu( \sigma \in f_v \mid Q_v(t,\sigma) = Q_v(t,\tau)  )  .
\end{align}
Indeed,  maximizing~\eqref{third_charge}  over $\tau \in \Omega $ yields~\eqref{goal} and completes the proof.

Fix $\tau = ( M_1, M_2, \ldots, M_N) \in \Omega$. To compute the sum on the left-hand side of~\eqref{third_charge} we need to determine the set of states $\mathrm{In}_v(\tau) \subseteq f_v$ for which  $\rho_{f_v}(\omega,\tau) > 0$.  To do this,  recall that given as input a state    $\omega = (M_1^{\omega}, M_2^{\omega}, \ldots, M_N^{\omega} ) \in f_v$,  procedure  {\sc Resample($ v , \omega, t )$}  modifies  one by one each matching $M_i$, $i \in [N]$, ``locally" around $v$.  In particular, observe that the support of the distribution for updating $M_i$ is exactly the set $\mathcal{M}_{t+1}^i(v,\omega)$, and hence it must be the case that $Q_v^i(t,\omega) = Q_v^i(t,\tau)$ for every $i \in [N]$ and state $\omega \in \mathrm{In}_v(\tau)$. This also implies that, for every such $\omega $,
\begin{align} \label{fourth_charge}
\frac{ \mu(\omega) }{ \mu(\tau) }   =  \prod_{i=1}^N \frac{ \nu(M_i^{\omega}) }{   \nu(M_i)}   = \prod_{i=1}^{N}  \frac{  \lambda (M_i^{\omega}  \cap E( G_{<t+1}(v) )  )  }{  \lambda (M_i \cap E(G_{<t+1}(v) ) ) }  .
\end{align}
Recall  now that we have assumed that the hard-core distribution  in Lines~\ref{seven},~\ref{eight} of { \sc Resample $(v,\omega,t)$ }is induced by the ideal vector of activities $\lambda$. In particular, we have
\begin{align}
\rho_{f_v}(\omega, \tau) &= \prod_{i= 1}^{N}  \frac{ \lambda (M_i \cap E(G_{<t+1}(v) ) ) }{ \sum_{ M \in \mathcal{M}_{t+1}^i(v , \omega ) } \lambda(M)} \nonumber \\
				  &   =\prod_{i= 1}^{N}  \frac{ \lambda (M_i \cap E(G_{<t+1}(v) ) ) }{ \sum_{ M \in \mathcal{M}_{t+1}^i(v,\tau)  } \lambda(M)}  
\end{align}
since $Q_v^i(t, \omega) = Q_v^i(t,\tau)$,  which combined with~\eqref{fourth_charge} yields
\begin{align}\label{five_charge}
\frac{ \mu(\omega) }{ \mu(\tau)  }  \rho_{f_v}(\omega, \tau) =   \prod_{i =1}^N   \frac{  \lambda (M_i^{\omega} \cap E(G_{<t+1}(v) ) )  }{ \sum_{ M \in \mathcal{M}_{t+1}^i(v,\tau)  } \lambda(M) }.
\end{align} 
Letting $\sigma = (M_1^{\sigma}, \ldots, M_N^{\sigma} )  $ be a random state distributed from $\mu$ we see that, by definition, the right-hand side of~\eqref{five_charge} equals:
\begin{align}\label{kouta}
\prod_{i =1}^N   \frac{  \lambda (M_i^{\omega} \cap E(G_{<t+1}(v) ) )  }{ \sum_{ M \in \mathcal{M}_{t+1}^i(v,\tau)  } \lambda(M) } = \prod_{i=1}^N \nu(  M_i^{\sigma} = M_i^{\omega}    \mid Q_v^i(t,\sigma) =  Q_v^i(t,\tau) )  = \mu( \sigma = \omega \mid Q_v(t,\sigma) =  Q_v(t,\tau)).
\end{align}
Finally, combining~\eqref{five_charge} and~\eqref{kouta}, we obtain that
\begin{align*}
\sum_{ \omega \in f_v } \frac{ \mu(\omega) }{ \mu(\tau)  }  \rho_{f_v}(\omega, \tau) = \sum_{\omega \in f_v}  \mu( \sigma = \omega \mid Q_v(t,\sigma) =  Q_v(t,\tau)) = \mu( \sigma \in f_v \mid    Q_v(t,\sigma) =  Q_v(t,\tau) ),
\end{align*}
concluding the proof of the first part of Lemma~\ref{termination}.

%Finally,   recall that $X_v = \{ x \in [0,1]^N: x = x_{v}(\omega) \text{ for some }  \omega \in f_v \} $, and specifically that $\omega \in f_v$ iff $x_{v}( \omega) \in X_v$.  For $x \in X_v$, let  $\Omega_{v,x} = \{ \omega: x_v(\omega) = x \}$. We now have
%\begin{eqnarray}
%\sum_{ \omega \in f_v  }  \frac{ \mu(\omega)  }{ \mu(\tau)  }  \rho_{f_v}(\omega, \tau)  &= & \sum_{x \in X_v }  \sum_{ \omega \in \Omega_{v,x}}  \frac{ \mu(\omega)  }{ \mu(\tau)  }  \rho_{f_v}(\omega, \tau)  \nonumber \\
%  &= &  \sum_{x \in X_v } \sum_{ \omega \in \Omega_{v,x}} \prod_{i =1}^N   \frac{  \lambda(M_i^{\omega} \cap E(G_{<t+1}(v) ) )  }{ \sum_{ M \in \mathcal{M}_{t+1}^i(v,\tau)  } \lambda(M) }   \label{deuteri}  \\
%& = & \sum_{x \in X_v }   \prod_{i =1}^N  \sum_{ \substack{  \omega \in \Omega_{v,x} \\ x_{v,i} (\omega) = x_i }}    \frac{  \lambda(M_i^{\omega} \cap E(G_{<t+1}(v) ) )  }{ \sum_{ M \in \mathcal{M}_{t+1}^i(v,\tau)  } \lambda(M) }   \label{triti} \\
% &=& \sum_{x \in X_v }   \prod_{i =1}^N \frac{       \sum_{ M \in \mathcal{M}_{t+1}^i(v,\tau), |M \cap E_v| = x_i}  \lambda(M)  }{   \sum_{ M \in \mathcal{M}_{t+1}^i(v,\tau)}  \lambda(M) }  \label{tetarti} .
%\end{eqnarray}
%To get~\eqref{deuteri}  we used~\eqref{five_charge}. For~\eqref{triti} we used the fact that $\Omega $ is the product space $\mathcal{M}(G)^N$, so that the choices per matching are independent,  while for~\eqref{tetarti}  we used the definition of $x_{v,i}(\omega)$.
%
%
%Combining~\eqref{tetarti} with~\eqref{first_charge} and~\eqref{second_charge} establishes~\eqref{third_charge}, concluding the proof of part~(a).
\end{proof}

\begin{proof}[Proof of part (b)]

For this proof we will use the following well-known proposition, which we also prove here for completeness.

\begin{proposition}\label{standard_counting_prop}
For every vertex $v$ there are at most $(\mathrm{e} \Delta)^{s-1}$ sets of vertices $S$ such that (i) $v \in S$; (ii) $|S| = s$; and (iii) $G[S]$ is connected.
\end{proposition}
\begin{proof}

The number of such sets is bounded by the number of distinct $s$-vertex trees which are rooted at $v$. The latter quantity is bounded by the number of distinct $\Delta$-ary rooted trees with $s$ vertices, which is 
\begin{align}\label{number_of_trees}
T_{\Delta}(s) := \frac{ { \Delta s \choose s } }{ (\Delta-1)s +1},
\end{align}
see e.g.~\cite{knuth1968art}. It is not hard to see that $T_{\Delta}(s) \le (\mathrm{e} \Delta)^{s-1}$ for $s \in \{1,2\}$ and $\Delta \ge 1$. For $s \ge 3$, we  obtain
\begin{align*}
T_{\Delta}(s) \le \frac{   \left( \frac{ \Delta s  \cdot \mathrm{e} }{ s } \right)^{s}   }{ (\Delta-1)s + 1 }  = \frac{   \left(  \Delta   \cdot \mathrm{e}  \right)^{s}   }{ (\Delta-1)s + 1 }  \le \left(  \mathrm{e}  \Delta \right)^{s-1 } 
\end{align*}
for sufficiently large $\Delta$, concluding the proof. Note that in deriving the first inequality we used that ${ a \choose b } \le (a \cdot \mathrm{e} / b)^{ b} $ for positive integers $b \le a$.

\end{proof}
To prove part (b) of Lemma~\ref{termination} it suffices to show that
\begin{align}\label{with_my_shot}
| F_{t+2}(v)|  \le \Delta^{ t + 2 \Delta^{1/3} + 3}
\end{align}
for every vertex $v$. Indeed,~\eqref{with_my_shot} clearly suffices if $f = f_v$. If $f = f_H$, notice that every $H$ for which we define $F_{t+2}(H)$ has fewer than $\Delta$ vertices  (assuming $\Delta$ is sufficiently large) and, therefore, every $F_{t+2}(H)$ has less than $D = \Delta^{ t+ 2\Delta^{ 1/3 }  +4} $ elements.

Towards proving~\eqref{with_my_shot}, at first notice  that  every set $S_{ < t+2}(v)$  has at most $\Delta^{t+2}$ elements. Moreover, using Proposition~\ref{standard_counting_prop}  we obtain that, for sufficiently large $\Delta$, every vertex $u$ is in at most 
\begin{align*}
 H_u := \sum_{s=1 }^{\frac{ 8\Delta }{ \epsilon N }    }  \left( \mathrm{e} \Delta \right)^{ s-1}    \le \frac{1}{ \mathrm{e} \Delta}  \cdot \frac{ \mathrm{ (e \Delta})^{  \frac{ 8 \Delta}{ \epsilon N } +1}  -1 }{ \mathrm{e} \Delta -1   }   \le \Delta^{ 2\Delta^{ 1/3 }  }
 \end{align*}
  sets~$H$ corresponding to a flaw~$f_H$. Note that in deriving the second inequality above we used the fact that $N = \lfloor \chi_e^*(G)^{3/4} \rfloor = \Theta(\Delta^{3/4} )$, which in turn implies that $\frac{ 8\Delta }{ \epsilon N } \le 2 \Delta^{ 1/3}$ for sufficiently large $\Delta$.   Overall:
  \begin{align*}
  | F_{t+2}(v)|  \le |S_{ < t+2}(v) |  \cdot \left(  \max_{u \in  S_{ < t+2}(v) } H_u +1 \right)  \le \Delta^{t+2} \cdot \left(\Delta^{ 2\Delta^{ 1/3 }  } + 1 	\right) \le \Delta^{ t + 2 \Delta^{1/3} + 3}
  \end{align*}
  for sufficiently large $\Delta$, concluding the proof.
\end{proof}

\subsection{Proof of Lemma~\ref{commutativity_lemma}}\label{commutativity_lemma_proof}

Fix states $\sigma_1 = (M_1, M_2, \ldots, M_N)  \in f$ and $\sigma_2 = (M_1', M_2', \ldots, M_N')  \in g$
such that $f\not\sim g$. To prove that the matrices $A_f, A_g$ commute, we need to show that for every such pair
\begin{align}\label{key_for_comm}
\sum_{ \tau} \rho_{f}(\sigma_1,\tau) \rho_{g}(\tau, \sigma_2) = \sum_{ \tau} \rho_{g}(\sigma_1,\tau) \rho_{f}(\tau, \sigma_2)  .
\end{align}
To that end, let $H_f, H_g$ be the subgraphs  (which may consist only of a single vertex) associated with flaws $f$ and $g$, respectively.  Since $f \nsim g $ we have $\min_{ u \in  V( H_f), v \in V(H_g) } \mathrm{dist}(u,v)  \ge t+2 $, where $\mathrm{dist}(u,v) $ denotes the length of the shortest path
%\marginpar{\tiny AS: Reinstated this as a formula rather than ``...is the empty set."  OK?}
between $u$ and $v$. Notice that this implies $S_{<t+2} (H_f ) \cap S_{<t+2} (H_g)  = \emptyset $.
%the intersection of sets $S_{<t+2} (H_f )$ and $S_{<t+2} (H_g)$ is the empty set.

Consider a pair of transitions $ \sigma_1 \xrightarrow{f} \tau$, $\tau \xrightarrow{g} \sigma_2$, where $\tau = ( M_1'', \ldots, M_N'')$,  and  so that  $\rho_f(\sigma_1, \tau) >0 $, $ \rho_g(\tau, \sigma_2) > 0$. The facts that  procedure {\sc Resample $(\sigma,f,t)$}  only modifies the input set of matchings
locally within $S_{<t+1}(H_f)$,   that $ \rho_g(\tau, \sigma_2) > 0$, and that $S_{<t+2} (H_f ) \cap S_{<t+2} (H_g)  = \emptyset $ imply that  (i)~$ \sigma_1 \in g$;  and (ii)~for every $i \in [N]$, $M_i   \cap ( S_{<t+2} (H_g)   ) =  M_i''   \cap ( S_{<t+2} (H_g)   ) $. 
%\AS{Changed notation for action set from $\alpha(\cdot,\cdot)$ to $a(\cdot,\cdot)$ for consistency with earlier. \red{OK. (It's very hard for me to notice the difference between alphas and a's.) }}
Notice now that the probability distribution $\rho_g(\tau, \cdot ) $  depends only on $(M_1'' \cap S_{<t+2}(H_g), \ldots, M_N'' \cap S_{t+2}(H_g)  )$. Hence, (i) and (ii)  imply  that the  probability distribution $\rho_g( \sigma_1 , \cdot ) $ is well defined and, in addition, there exists a natural bijection~$b_g$ between the action set $a(g,\tau)$ and the action set $a(g, \sigma_1)$  so that $\rho_g(\tau, \tau' ) = \rho_g(\sigma_1,  b_g(\tau'))$ for every $\tau' \in a(g,\tau)$. This is because both distributions are implemented by sampling from the set of matchings of the same multigraph according to the same probability distribution.

Now let $ \tau' = b_g(\sigma_2)  $. A symmetric argument implies that $\tau'\in f$ and that there exists a natural bijection~$b_f$ between $a(f, \sigma_1 ) $ and $a(f, \tau' )$ so that  $\rho_f(\sigma_1, \sigma)  = \rho_{f}(\tau', b_f(\sigma)) $ for every $\sigma \in a(f, \sigma_1)$. 
%\AS{Changed the derivation below slightly; plse check carefully!!! \red{OK}}
In particular, notice that  $\sigma_2 =  b_f(\tau)$ and that
\begin{align}
\rho_f(\sigma_1, \tau) \rho_g(\tau, \sigma_2) & =   \rho_g(\sigma_1, \tau') \rho_f(\tau', b_f(\tau))  \nonumber  \\
							           &  = \rho_g(\sigma_1, \tau') \rho_f( \tau', \sigma_2 ) .  \label{product_holds}
%\rho_f(\sigma_1, \tau) \rho_g(\tau, \sigma_2)  =   \rho_g(\sigma_1, b_g(\sigma_2)) \rho_f(b_g(\sigma_2), b_f(\tau))     = \rho_g(\sigma_1, \tau') \rho_f( \tau', \sigma_2 ) .  
\end{align}
Overall, what we have shown is a bijective mapping that sends any pair of transitions $\sigma_1 \xrightarrow{f} \tau, \tau \xrightarrow{g} \sigma_2$  to a pair of transitions $\sigma_1 \xrightarrow{g} \tau', \tau' \xrightarrow{f} \sigma_2$ and which satisfies~\eqref{product_holds}. This establishes~\eqref{key_for_comm}, concluding the proof. \qquad$\square$

\section{List-Edge Coloring Multigraphs: Proof of Theorem~\ref{main_list}}\label{list_main}
In this section we review the proof of Theorem~\ref{list_kahn} and then prove its
constructive version, Theorem~\ref{main_list}.  \blue{Again, throughout the proof we assume that the maximum degree $\Delta$ of  the input multigraph $G$ satisfies $\Delta \ge \Delta_0$ for some appropriately large constant $\Delta_0$.}

In Section~\ref{royals} we give a high-level sketch of the existential proof of Kahn, and we state the key technical results from that paper (Theorems~\ref{first_th},~\ref{second_th}, and Lemma~\ref{key_list_kahn_lemma}). As we will see, our main contribution is to make Theorem~\ref{first_th} constructive. Towards this end, we describe our local search algorithm in Section~\ref{list_algo}, where we also prove its correctness assuming it converges (Lemma~\ref{avoiding_fe_suffices}), as well as an important property of the flaws we consider (Lemma~\ref{key_list_kahn_lemma_new}).
Finally, in Section~\ref{proofara} we prove that our search algorithm  has  expected polynomial running time, concluding the proof of Theorem~\ref{main_list}.

\subsection{A High Level Sketch of the Existential Proof}\label{royals} 

 As we explained in the introduction, the non-constructive proof of Theorem~\ref{list_kahn}  is a sophisticated version of the semi-random method and proceeds by partially coloring the edges of the  multigraph in iterations, until at some point the coloring can be completed greedily. (More accurately, the method establishes the \emph{existence} of such a sequence of desirable partial colorings.)

We will  follow the exposition in~\cite{mike_book}. In each iteration, we have a list $L_e$ of acceptable colors for each edge $e$. We assume that each $L_e$ originally has $C$ colors for some $C \ge (1+\epsilon) \chi_e^*(G)$, where $\epsilon >0$ is an arbitrarily small constant. For each color  $i$, we let $G_i$ be the subgraph of $G$ formed by the edges for which $i$ is acceptable. Since $G_i \subseteq G, \chi_e^*(G_i) \le \chi_e^*(G)$. Thus, Theorem~\ref{hardcore_dist} implies that we can find a hard-core distribution on the matchings of $G_i$ with marginals $(\frac{1}{C}, \ldots, \frac{1}{C} )$ whose activity vector $\lambda_i$ satisfies $\lambda_i(e) \le \frac{K}{C}$ for all $e$, where $K= K(\epsilon)$ is a constant.

In each iteration, we will use the  \emph{same} activity vector $\lambda_i$ to generate the random matchings assigned to color $i$. Of course, in each iteration we restrict our attention to the subgraph of $G_{i}$ obtained by deleting the set $E^*$ of edges colored (with any color) in previous iterations,  and the endpoints of the  set of edges  $E_{i}^*$ colored $i$ in previous iterations. (Thus,  although we use the same activity vector for each color  in each iteration, the induced hard-core distributions may vary significantly.)   Further, we will make sure that our distributions have the property that for each edge $e$, the expected number of matchings containing~$e$ is very close to~$1$. \blue{(In other words, the sum over $i$ of the probabilities that edge $e$ is a part of the matching corresponding to color $i$ is close to~$1$.)}

We apply the Lopsided LLL in the following probability space. For each color $i$, independently, we choose a matching $M_i \in G_i$  from the corresponding distribution. Next, we \emph{activate} each edge in $M_i$ independently with probability $\alpha := \frac{1}{ \log \Delta(G) } $; we assign colors only to activated edges in order to ensure that very few edges are assigned more than one color. We then update the multigraph by deleting the colored edges, and update the lists $L_e$ by deleting any color assigned to an edge incident on~$e$.   We give a more detailed description below.

Notice that our argument needs to ensure that (i) at the beginning of each iteration the induced hard-core distributions are such that, for each uncolored edge $e$, the expected number of random matchings containing $e$ is very close to $1$; and (ii)  after some number of iterations, we can complete the coloring greedily. 

As far as the latter condition is concerned, notice that if (i) holds throughout then, in each iteration, the probability that an edge retains a color remains close to the activation probability $\alpha$. This allows us to prove that the maximum degree in the uncolored multigraph drops by a factor of about $1-\alpha$ in each iteration. Hence, after $\log_{\frac{1}{1-\alpha} } 3K$ iterations, the maximum degree in the uncolored multigraph will be less than $\frac{\Delta}{2K} $. Furthermore, for each $e$ and $i$, the probability that $e$ is in the random matching of color $i$ is at most $\lambda_i(e) \le \frac{K}{C}$. Since (i)
continues to hold, this implies there are at least $\frac{C}{K} > \frac{\Delta}{K}$ colors available for each edge, and so the coloring can be completed greedily. (Recall that the $C >\chi_e^*(G) \ge \Delta$.)

\par\bigskip
{\bf An Iteration.}
\begin{enumerate}

\item For each color $i$, pick a matching $M_i$ according to a hard-core probability
%\AS{Is $\mu_i$ the right notation here, or should it be $\nu_i$ for consistency with earlier???}
distribution $\mu_i$ on $\mathcal{M}(G_i)$ with activities $\lambda_i$ such that for some constant $K$:

\begin{enumerate}[(a)]

\item $\forall e \in E(G), \sum_{ i } \mu_i ( e \in M_i ) \approx 1$; and
%\marginpar{\tiny AS: Added additional $\forall$ quantifier for edges in (b); please check!!!}

\item  $\forall i , \forall e \in E(G), \lambda_i(e) \le \frac{K}{C} $ and hence $\forall v \in V(G), \sum_{L_e \ni i } \lambda_i(e) \le K$.

\end{enumerate}

\item For each $i$, activate each edge of $M_i$ independently with probability $\alpha = \frac{1}{ \log \Delta(G)}$, to obtain a new matching $F_{i} \subseteq M_i$. We color the edges of $F_i$ with color $i$ and delete $V(F_i)$  from $G_i$. We also delete from $G_i$ every edge not in $M_i$ which is in $F_{j} $ for some $j \ne i$. We do not delete edges of $(M_i - F_i) \cap F_{j} $ from $G_i$. (Note that this may result in edges receiving more than one color, which is not a problem since we can always pick one of them arbitrarily at the end of the iterative procedure.) 

\item  Perform an \emph{equalizing coin flip} for each edge $e$ of $G_i  $  so that the probability  that $e$ is both colored and removed from $G_i$ in either Step 2 or Step 3 is exactly $\alpha$.  (See also Remark~\ref{eq_flips} below.)

\end{enumerate}

\blue{
\begin{remark}\label{eq_flips}
 Note that the expected number of edges that are both colored and removed from $G_i$  in Step 2 is less than $\alpha | E(G_i) |$ because, although the expected number of colors retained by an edge is very close to $\alpha$,  some edges may be assigned  more than one color.  Performing ``equalizing coin flips" in Step 3 is a standard  idea  that helps in avoiding several technical difficulties that stem from the latter fact.
\end{remark}
}

The \emph{outcome} of an iteration is defined to be the choices of matchings, activations, and equalizing coin flips. Let $\mathrm{Out}= \mathrm{Out}_{\ell}$ denote the random variable that equals the outcome of the $\ell$-th iteration. (In what follows, we will focus on a specific iteration $\ell$ and so we will omit the subscript.)

For each edge $e =(u,v)$, we define a bad event $A_e$ as follows.  Let $G_i'$ be the multigraph obtained after carrying out the modifications to $G_i$  in Steps 2 and  3 of the above iteration.  Let $t = 8(K+1)^2 (\log \Delta)^{20} + 2$, recall the definition of $S_{<t}( H) $ for subgraph $H$, and let $G_{<t}(H)$ denote the corresponding induced subgraph.  Let $Z_i$ be a random matching in $G_i' \cap G_{<t}(\{u,v \}) $ sampled from  the hard-core probability distribution induced by  activity vector $\lambda_i$. Let $A_e$ be the event that\\
\begin{equation}\label{edge_flaw}
\Bigl| \sum_{i: G_i' \ni e  } \Pr ( e \in Z_i \mid \mathrm{Out}  ) - \sum_{i: G_i \ni e } \Pr(  e \in M_i )   \Bigr|  >  \frac{1 }{ 2 (\log \Delta)^4 }  .
\end{equation}
%\begin{align}\label{edge_flaw}
%&\Bigl| \sum_{i: G_i' \ni e  } \Pr ( e \in Z_i \mid Q  ) - \sum_{i: G_i \ni e } \Pr(  e \in M_i )   \Bigr|  \nonumber \\ > & \frac{1 }{ 2 (\log \Delta)^4 }  .
%\end{align}
%\begin{align}\label{edge_flaw}
%&\Bigl| \sum_{i: G_i' \ni e  } \Pr ( e \in Z_i \mid Q  ) - \sum_{i: G_i \ni e } \Pr(  e \in M_i )   \Bigr| >  \frac{1 }{ 2 (\log \Delta)^4 }.
%\end{align}

To get some intuition behind the definition of event $A_e$, let $M_i'$ be a random matching in $G_i'$ chosen according to the hard-core distribution  with activities $\lambda_i$. Since correlations decay with distance, one can show that $\Pr( e \in M_i' \mid \mathrm{Out})$ is within a factor $1 + \frac{1}{ (\log \Delta )^{20}}  $  of $\Pr( e \in Z_{i} \mid \mathrm{Out} ) $. Thus, according to~\eqref{edge_flaw},  avoiding bad event $A_e$ implies that $ \sum_{i} \Pr( e \in M_i'  ) \approx \sum_{i} \Pr ( e \in M_i ) \approx 1 $, which is what is required in order to maintain property (i) at the beginning of the next iteration. In particular, it is straightforward to see that avoiding all bad events $\{A_e \}_{e \in E(G) } $ guarantees that \\
\begin{equation}\label{edge_guarantee}
\Bigl| \sum_{i: G_i' \ni e  } \Pr ( e \in M_i' \mid \mathrm{Out} ) - \sum_{i: G_i \ni e } \Pr(  e \in M_i )   \Bigr|  \le  \frac{1 }{  (\log \Delta)^4 }  ,
\end{equation}
%\begin{align}\label{edge_guarantee}
%&\Bigl| \sum_{i: G_i' \ni e  } \Pr ( e \in M_i' \mid Q  ) - \sum_{i: G_i \ni e } \Pr(  e \in M_i )   \Bigr| \nonumber\\  \le & \frac{1 }{  (\log \Delta)^4 }  ,
%\end{align}
for sufficiently large $\Delta$, which is what we really need.
The reason we consider $Z_{i}$ and not $M_i'$ is that  events defined with respect to the former are mildly negatively correlated with most other bad events, making it possible  to apply the Lopsided LLL.

Further, for each vertex $v$ we define $A_v$  to be the event that the proportion of edges incident on~$v$ which are colored in the iteration is less than $\alpha - \frac{1}{  (\log \Delta)^4} $.

It can be formally shown that, if we avoid all bad events, then (i) holds, i.e., at the beginning of the next iteration we can  choose new probability distributions so that for each uncolored edge $e$ we maintain the
property that the expected number of random matchings containing $e$ is  very close to $1$, and, moreover, after $\log_{ \frac{1}{ 1- \alpha }  }3K $ iterations we can complete the coloring greedily.

\begin{theorem}[\cite{kahnListChrom}]\label{first_th}
Assume that~\eqref{edge_guarantee} holds for the edge marginals of the matching distributions of iteration  $\ell$. Then, with positive probability, the same is true for the matching distributions of iteration $\ell+1$.
\end{theorem}

\begin{theorem}[\cite{kahnListChrom}]\label{second_th}
If we can avoid the bad events of the first  $\log_{ \frac{1}{ 1- \alpha }  }3K $ iterations, then we can complete the coloring greedily.
\end{theorem}

Proving Theorems~\ref{first_th} and~\ref{second_th}  is the heart of the proof of Theorem~\ref{list_kahn}. 
The most difficult part is proving that, for any $x \in V \cup E$, the probability of event $A_x$ is very close to $0$ conditioned on any choice of outcomes for distant events. (This is needed in order to apply the Lopsided LLL.)  
%Given Theorem~\ref{first_th}, and as we have already explained, the proof of Theorem~\ref{second_th} follows from the fact that in each iteration the expected number of random matchings containing each uncolored edge $e$ is very close to $1$ and, therefore, the probability that $e$ retains  a color remains close to $\alpha$.  
Given Theorem~\ref{first_th}, the proof of Theorem~\ref{second_th} follows, as we have already explained, from the fact that in each iteration the expected number of random matchings containing each uncolored edge $e$ is very close to $1$ and, therefore, the probability that $e$ retains  a color remains close to $\alpha$.  

Below we state the key lemma that is proven in~\cite{kahnListChrom}, and which we will also use in the analysis of our algorithm.

Recall the definition of $t$. For a subgraph $H$, we let $R_H$ be the random outcome of our iteration in $G - S_{<t^2}(H)$, i.e., $R_H$ consists of $ \bigcup_{i} \left( M_{i} - S_{< t^2}( H )  \right)$ together with the choices of the activated edges in $G - S_{<t^2 }( H)$ which determine the $\bigcup_{i} \left( F_i - S_{<t^2}( H) \right)$, and the outcomes of the equalizing coin flips for edges in this subgraph.

\begin{lemma}[\cite{kahnListChrom}]\label{key_list_kahn_lemma} 
For every  $x \in E \cup V$ and possible choice $R_x^*$ for $R_x$,   we have $\Pr(A_x \mid R_x= R_x^*) \le  \frac{1}{  \Delta^{3(t^2+t+2) }  } $.
\end{lemma}
  
In the remaining sections we will focus on providing an efficient algorithm for  Theorem~\ref{first_th} which, combined with Theorem~\ref{second_th}, will imply the proof of Theorem~\ref{main_list}. 

%
%As a final remark, we note that detecting whether bad events $\{ A_e \}_{e \in E(G) }$ are present in a state is not a tractable task since it entails the exact computation of edge marginals of hardcore distributions over matchings. In order to overcome this obstacle, we will define flaws $\{f_e \}_{e \in E(G) }$ whose absence  provides somewhat weaker guarantees than the removal of bad events $\{A_e \}_{e \in E(G) }$, but nonetheless implies~\eqref{edge_guarantee} for every edge. 
%To decide whether a flaw $f_e$ is present in a state, we will use the results of~\cite{Alistair1} to estimate the corresponding edge marginals of random variables $M_i$ and $Z_i$ for every color~$i$. Note that since we will only perform an approximation, it is possible that we deduce that $A_e$ is not present while in reality it is. However,   our approximation will be tight enough so that, even  in this case,~\eqref{edge_guarantee} will still hold for every edge. We give the details below.
%
%
%

As a final remark, we note that detecting whether bad events $\{ A_e \}_{e \in E(G) }$ are present in a state is not a tractable task since it entails the exact computation of edge marginals of hardcore distributions over matchings. In order to overcome this obstacle, we will define flaws $\{f_e \}_{e \in E(G) }$ whose absence  provides somewhat weaker guarantees than the removal of bad events $\{A_e \}_{e \in E(G) }$, but nonetheless implies~\eqref{edge_guarantee} for every edge. 
To decide whether a flaw $f_e$ is present in a state, we will use the results of~\cite{Alistair1} to estimate the corresponding edge marginals of random variables $M_i$ and $Z_i$ for every color~$i$. Note that since we will only perform an approximation, we will not be able to check for~\eqref{edge_flaw} directly. However,   our approximation will be tight enough so that, even  in this case,~\eqref{edge_guarantee} will still hold for every edge. We give the details below.

\subsection{The Algorithm}\label{list_algo}

Let $\mathcal{U} $ denote the set of uncolored edges and $N = |\bigcup_{ e \in \mathcal{U}} L_e|$ ,  the cardinality of the set of colors that appear in the list of available colors of some uncolored edge. For a color $i \in [N]$, recall that $G_i$ denotes the subgraph of uncolored edges that contain $i$ in their list of available colors. Finally, let $E_i = | E(G_i) | $ and $\mathcal{S} = \mathcal{S}(T)$ be the set of all binary strings of length $T$, where $T$ is a  parameter to be defined later. An element of $\mathcal{S}$ should be thought of  as the input ``randomness" to a subroutine of our algorithm whose purpose will be to estimate edge-marginals of distributions $\{ \mu_i \}_{i \in [N] } $. 

Define  $\Omega = \prod_{ i  \in [N] } \left(   \mathcal{M}(G_i)  \times   \{0,1 \}^{E_i } \times   \{0,1 \}^{E_i } \times    \mathcal{S}^{E_i} \right)$.
We consider an arbitrary but fixed ordering over  $\mathcal{U}$, so that each state $\sigma \in \Omega$ can be represented as  $\sigma =  \left(  (M_1,a_1, h_1,s_1), \ldots,   (M_N, a_N, h_N,s_N)  \right)  $, where $M_i,a_i, h_i$ are the matching,  activation and equalizing coin flip vectors, respectively, that correspond to color $i$,  so that  edge $e$ is activated in $G_i$ if $a_i(e) = 1$ and is marked to be removed if $h_i(e)=1$. Additionally, $s_i$ is the tuple of strings corresponding to the particular element of $\mathcal{S}^{E_i}$ at state $\sigma$. As we will see, tuples $\{s_i \}_{i \in [N] } $ are defined  for purely technical reasons, and specifically for properly bypassing the issue of detecting the presence of $f_e$-flaws that we mentioned earlier. 

\blue{Recalling  Corollary~\ref{easy_cor}, we see that we are able to  obtain a $1 \pm 1/n^{\beta}$ approximation for the marginal $\mu_i(e)$, $i \in [N]$,  of an edge $e$ with probability at least $1 - 1/n^{\beta}$ in polynomial time, where $\beta$ is a fixed and sufficiently large positive constant. This fact will be useful to us in two ways.}

First, recall that for color  $i$ we choose a matching according to probability distribution $\mu_i$, and we define $\mathrm{Eq}_i(e)$ to be the probability of success of the equalizing coin flip that corresponds to  edge $e$ and color $i$. Note that, given access to the marginals of $\mu_i$, the value of $\mathrm{Eq}_i(e)$ can be computed efficiently. Of course, and as we just explained,  we will have only (arbitrarily good) estimates of the marginals of $\mu_i$, but as in the proof of Theorem~\ref{main}, this suffices for our purposes. Indeed, through sampling we can  efficiently get an estimate $\mathrm{Eq}_i'(e)  $ that is within a  $1 \pm 1/n^{c} $ factor of the correct value~$\mathrm{Eq}_i(e)$ with probability at least $1- 1/n^c$, where $c = c(\beta)$ is a sufficiently large constant,  and hence guarantee that the total variation distance between the resampling probability distributions used by the algorithm and the ideal ones is negligible, i.e.,  at most $1/n^{c}$.  (Later we will argue that we can  maximally couple the approximate and ideal distributions and proceed with an argument identical to the one we used in the proof of Theorem~\ref{main}, where we absorb the probability that the coupling fails into the expected polynomial running time of the algorithm — recall our discussion in the beginning of Section~\ref{natooos})

Second, we let $T_1 = T_1(\beta) = \mathrm{poly}(n) $ be a fixed polynomial upper bound on the number of random bits required by the sampling algorithm  (whose existence is guaranteed by Theorem~\ref{sampling_matchings}) for approximating $\Pr(e \in M_i)$, for an arbitrary color $i \in [N]$ and an arbitrary edge $e$, within a factor $1 \pm 1/n^{\beta}$ with probability at least $1- 1/n^{\beta}$. We let $T_2$ be an analogous fixed polynomial upper bound for estimating $\Pr( e \in Z_i \mid \mathrm{Out} ) $ for arbitrary $\mathrm{Out}$, and define $T = T_1 + T_2$.

We let $p$ be the probability distribution over $\Omega$ that is induced by the product of the $\mu_i$'s, activation flips,  equalizing coin flips, and the uniform distribution over $\mathcal{S}^{E_i}$, for each color $i \in [N]$. In other words,  $p$ is the probability distribution over $\Omega$ induced by the iteration along with some extra randomness that is used for sampling from $\mathcal{S}(T)^{E_i}$.

\paragraph{The initial distribution.}  Recall that each edge $e$ initially has a list $L_e$ of size at least $(1+\epsilon) \chi_e^*(G) $.   As we have already seen in Corollary~\ref{computing_the_dist}, the results of \cite{Alistair1,singh2014entropy} imply that for every color $i$ and parameter $\eta = 1/n^{\beta}$, where $\beta >0$ is a sufficiently large constant,  there exists a  $\mathrm{poly}(n,  \ln \frac{1}{ \epsilon} )$-algorithm that computes a vector $\lambda'_i$  such that the induced hard-core distribution $\eta$-approximates in variation distance the hard-core distribution induced by vector $\lambda_i$. 
Let  $p'$ be the distribution obtained in an identical way to  $p$ but using vectors $\lambda'_i$ instead of vectors $\lambda_i$. The initial distribution $\theta$ of our algorithm is obtained by $\eta$-approximately sampling from $p'$. Theorem~\ref{sampling_matchings} implies that this can be done in polynomial time.

\paragraph{Finding and addressing  flaws.} 
We define a flaw $f_v$ for each bad event $A_v$. To define flaw $f_e$ corresponding to an edge $e$, we first recall the definitions of $T_1, T_2$.  In particular, recall that the description of a state $\sigma$ determines  a binary string $s  = s(\sigma) \in \mathcal{S}$ of length $T_1 +T_2$ for each color $i$ and  edge $e \in E(G_i)$.  We will think of $s$ as a concatenation of two strings of length  $T_1$ and $T_2$, respectively,  that can and  will be used  as the ``input randomness" to a sampling algorithm that estimates  $\Pr( e \in M_i )$ and $\Pr( e \in Z_i \mid \mathrm{Out}(\sigma)) $, respectively.  (Here $\mathrm{Out}(\sigma)$  is the evaluation of random variable $\mathrm{Out}$ at $\sigma$.) Indeed,  let $\widetilde{\Pr_{\sigma}}(  e \in M_i )$ be the resulting, \emph{deterministic} (given $s(\sigma) $) estimation of $\Pr( e \in M_i)$ and, similarly, let $\widetilde{ \Pr_{\sigma} } ( e \in Z_i \mid \mathrm{Out}(\sigma)  )$ be the  resulting   estimation of $\Pr( e \in Z_i \mid \mathrm{Out}(\sigma)) $. Finally, we define flaw $f_e$ to be the set of states $\sigma \in \Omega$ such that\\ 
%\begin{align}\label{edge_new_flaw}
%&\Bigl| \sum_{i: G_i' \ni e  } \Pr ( e \in Z_i \mid \sigma  ) - \sum_{i: G_i \ni e } \Pr(  e \in M_i )   \Bigr|  \nonumber\\
%&>  \frac{2 }{ 3 (\log \Delta)^4 }  .
%\end{align}
\begin{equation}\label{edge_new_flaw}
\Bigl| \sum_{i: G_i' \ni e  } \widetilde{ \Pr_{\sigma} } ( e \in Z_i \mid \mathrm{Out}(\sigma)  ) - \sum_{i: G_i \ni e } \widetilde{\Pr_{\sigma}}(  e \in M_i )   \Bigr|  
>  \frac{2 }{ 3 (\log \Delta)^4 }  .
\end{equation}
We fix an arbitrary ordering $\pi$ over $V \cup E$. In each step, the algorithm finds the lowest indexed flaw according to~$\pi$ that is present in the current state and addresses it. 

Clearly, checking if  vertex-flaws~$f_v$ are present in the current state can be done efficiently.  The same is true for edge-flaws~$f_e$ given Theorem~\ref{sampling_matchings}. What is perhaps not so clear, however, is whether the definition of $f_e$-flaws is sufficient for our purposes, and how it relates to the definition of bad events $A_e$.

To address these questions, recall first that we can use the results of~\cite{Alistair1} to approximate the edge marginals of  the corresponding distributions within a $(1\pm \eta )$-factor  with probability at least $1- \eta $, in time $\mathrm{poly}(n,  \frac{1}{\eta} )$, where $\eta = 1/n^{\beta}$. Our approach will be to  first argue that, assuming our edge marginal estimates were \emph{always} within a $(1\pm \eta)$-factor of the true values, then our algorithm would terminate in expected polynomial time, and  then use a coupling argument similar to the one described in the beginning of Section~\ref{natooos} to show that we can make this assumption in our analysis at a negligible price.

More formally, given a state $\sigma =  \left(  (M_1,a_1, h_1,s_1), \ldots,   (M_N, a_N, h_N,s_N)  \right) $, let $ M(\sigma) = (M_1, \ldots, M_N)$, $a(\sigma) = (a_1, \ldots, a_N)$, and $h(\sigma) = (h_1, \ldots, h_N)$, and define $\xi(\sigma) =  (M(\sigma), a(\sigma), h(\sigma) ) $. For each edge $e$, color $i$, and state $\sigma$, let  $\mathcal{S}_i'(e) =    \mathcal{S}_i'(e,\xi(\sigma))  \subseteq \mathcal{S}$  be the set of strings with the property that, if our marginal estimators use them as input randomness in state $\sigma$ for edge $e$, then they are guaranteed to provide a $(1\pm\eta)$-factor approximation of the true marginals of $e$.  Crucially, observe that   $| \mathcal{S}_i'(e) |/ |\mathcal{S} |  \ge 1-1/n^{c} $ for a  constant $c = c(\beta) $ which can be made arbitrarily large by increasing $\beta$.
Let $\Omega'  \subseteq \Omega$ be the subspace of~$\Omega$ induced by removing every state $\sigma =  \left(  (M_1,a_1, h_1,s_1), \ldots,   (M_N, a_N, h_N,s_N)  \right)  $ such that there exists an $i \in [N]$ for which $s_i \notin \prod_{e \in E(G_i)} \mathcal{S}_i'(e)$. That is, $\Omega'$ is the subspace of $\Omega$ in which our edge-marginal approximations are guaranteed to be within a $(1 \pm \eta)$-factor of the true values. Finally,  let $\mu$ be the distribution induced by conditioning on the event that a sample from $p$  belongs to~$\Omega'$.  Equivalently, to take a sample from $\mu$ we  first sample from the product of the $\mu_i$'s, activation flips,  and equalizing coin flips to obtain a tuple $\xi = (M,a,h)$,  and then sample uniformly an element from $\prod_{i=1}^N \prod_{e \in E(G_i) } \mathcal{S}_i'(e,\xi ) $.

The following two lemmas justify our definition of $f_e$-flaws. Specifically, Lemma~\ref{avoiding_fe_suffices} shows that avoiding all $f_e$-flaws is sufficient for the purposes of our analysis (recall Theorem~\ref{first_th}), 
while Lemma~\ref{key_list_kahn_lemma_new}  bounds the probability of each  flaw (with respect to $\mu$).

\begin{lemma}\label{avoiding_fe_suffices}
Condition~\eqref{edge_guarantee} holds for every edge $e$ and every state $\sigma \in \Omega'$ such that $\sigma \notin f_e$.
\end{lemma}
\begin{proof}
Since for every state $\sigma \in \Omega'$ we have that $\widetilde{ \Pr_{\sigma} } ( e \in Z_i \mid \mathrm{Out}(\sigma)  ), \widetilde{\Pr_{\sigma}}(  e \in M_i )    $ are within a $(1 \pm \eta)$ factor of the respective true marginals, we have that for every state $\sigma  \in \Omega' \setminus f_e$:
\begin{align*}
\Bigl| \sum_{i: G_i' \ni e  }\Pr ( e \in Z_i \mid \mathrm{Out}(\sigma)  ) - \sum_{i: G_i \ni e } \Pr(  e \in M_i )   \Bigr|  \le   \frac{2 }{ 3 (\log \Delta)^4 } + 2 \eta.
\end{align*}

Recalling that $\Pr( e \in M_i' \mid \sigma ) $ is within a  $(1 + \frac{1}{  (\log \Delta)^{20}} )$-factor  of $\Pr(   e \in Z_i \mid \sigma ) $, we can  deduce that if  flaw $f_e$ is not present in a state $\sigma \in \Omega'$, then~\eqref{edge_guarantee} holds for sufficiently large $\beta,\Delta$, as claimed. \end{proof}

\begin{lemma}\label{key_list_kahn_lemma_new} 
For every  $x \in E \cup V$ and possible choice $R_x^*$ for $R_x$ we have $\mu(f_x \mid R_x= R_x^*) \le  \frac{1}{  \Delta^{3(t^2+t+2) }  } $.
\end{lemma}
\begin{proof}
\blue{For $f_v$ flaws the claim follows almost immediately from Lemma~\ref{key_list_kahn_lemma},  so we focus on  proving it for the case of $f_e$-flaws. In particular, we  show that 
\begin{align}\label{this_very_lemma_goal}
\mu(f_e \mid R_e= R_e^*) \le  \Pr(A_e \mid R_e= R_e^*  ) 
\end{align}
as this implies our claim per Lemma~\ref{key_list_kahn_lemma}.

Recall that $f_e$ is a subset of $\Omega$, i.e., the ``augmented" space where each state is associated with a tuple of strings from $\prod_{i \in [N]}  \mathcal{S}^{E_i}$, while event $A_e$ is a subset of the original probability space that is induced by the family of random matchings, activations, and equalizing coin flips for each edge. Recall also that, by definition, $\mu$ assigns zero probability mass to $f_e \setminus \Omega'$, i.e., the part of $f_e$ where we have no guarantees about the quality of approximation of our edge-marginal estimators.  In order to establish~\eqref{this_very_lemma_goal} we ``project" $f_e \cap \Omega'$ to the original probability space to get an event $\widetilde{A}_e$. That is, the elements of $\widetilde{A}_e$ are induced by the elements of $f_e$ by ignoring the coordinate that corresponds to the tuple of strings from $\prod_{i \in [N]}  \mathcal{S}^{E_i}$.  By definition, $\mu( f_e \mid R_e = R_e^* ) = \Pr( \widetilde{A}_e \mid R_e = R_e^* )$. 

In addition, for every elementary event $\xi \in \widetilde{A}_e$ we have
\begin{align}\label{Afact}
\Bigl| \sum_{i: G_i' \ni e  }\Pr ( e \in Z_i \mid \mathrm{Out}(\xi) ) - \sum_{i: G_i \ni e } \Pr(  e \in M_i )   \Bigr|   \ge   \frac{2 }{ 3 (\log \Delta)^4 } - 2 \eta  > \frac{1 }{ 2 (\log \Delta)^4 },
\end{align}
for sufficiently large $\Delta$. Note that the first inequality follows from~\eqref{edge_new_flaw} and the fact that we only consider elements in $f_e \cap \Omega'$, i.e., states in which our edge-marginal approximations are within a $(1\pm \eta)$-factor from the true values.
Recalling~\eqref{edge_flaw}, we see that inequality~\eqref{Afact} implies that $\Pr( \widetilde{A}_e \mid R_e = R_e^* ) \le \Pr(A_e \mid R_e = R_e^*)  $ (and, therefore, also~\eqref{this_very_lemma_goal}), concluding the proof.

}

\end{proof}

Summarizing, we may assume without loss of generality that we are able to accurately and efficiently search for edge-flaws~$f_e$, and that their probability with respect to measure $\mu$ is bounded above by $\Delta^{-3(t^2+t+2)}$ conditional on any instantiation of $R_e$.

Recall the procedure {\sc Resample} described in Section~\ref{chrom_algo}. Below we describe procedure { \sc Fix } that takes as input a subgraph $H$ and a state $\sigma$. In the description of { \sc Fix }  below we invoke procedure { \sc Resample } with an extra parameter, namely an activity vector $\lambda_i'$  for each color $i$. By that we mean that in  Lines~\ref{seven},~\ref{eight} of { \sc Resample } we use the vector $\lambda_i'$ to define $p$. Finally, recall that we defined $t = 8(K+1)^2 (\log \Delta )^{20} + 2$.
%\red{Maybe ewe should explain what $t$ is again.} 

\begin{algorithm}
\begin{algorithmic}[1] 
\Procedure{Fix}{$H,\sigma$}
\State  Let $\sigma =  \left( (M_1,b_1, h_1,s_1),   \ldots, (M_N, b_N,h_N,s_N)  \right)$
\State  $(M_1',\ldots, M_N' )\leftarrow \text{{ \sc  Resample}}(H,(M_1, \ldots, M_N) ,t^2,\lambda_{i}')$ \label{first}
\For{$i = 1$ to $N$}
\State Update $a_i$ to $a'_i$ by activating independently each edge in $  G_{i,<t^2+1}( H )$ with probability $\alpha$ \label{second}
\State Update $h_i$ to $h'_i$ by flipping the  equalizing coin corresponding to each edge in  $ G_{i,<t^2+1}( H )$ \label{third}
\State Update $s_i$ to $s_i'$ by uniformly sampling from  $\mathcal{S}^{E_i}$ \label{fourth}
\EndFor
\State Output $\sigma =  \left( (M_1',a_1', h_1',s_1') , \ldots, (M_N', a_N',h_N', s_N')  \right)$	
\EndProcedure
\end{algorithmic}
\end{algorithm}

Theorem~\ref{sampling_matchings} implies that  procedure  {\sc Fix} runs in polynomial time for any input subgraph $H$ and state $\sigma$. To address flaws $f_v, f_{ \{ u_1, u_2 \} } $ in  a state $\sigma$  we invoke {\sc Fix($ \{ v \},\sigma$)} and    {\sc Fix($  \{u_1, u_2 \}   ,\sigma$)}, respectively.

\subsection{Proof of Theorem~\ref{main_list}}\label{proofara}

%\red{
%
%
% Given this fact, we can take $\beta$ to be a sufficiently large constant  and  use a coupling argument similar to the one we used in the proof of Theorem~\ref{main} (recall also Remark~\ref{coupling_remark}),   which  allows us to assume that our estimates are indeed always within a $(1+\eta)$-factor from the true values, and subsume the error probability into the expected running time of the algorithm.}

Similarly to the proof of Theorem~\ref{main}, for our analysis we will assume that  our algorithm samples from the ``ideal"  matchings distributions, i.e., the ones induced by the vectors $\lambda_i$ rather than by the approximations~$\lambda_i'$. We will also assume that each equalizing coin corresponding to a color $i \in [N]$ and an edge $e$ is flipped with probability of success $\mathrm{Eq}_i(e)$ instead of $\mathrm{Eq}_i'(e)$, and that we update string $s_i(e)$ by sampling uniformly from  $\mathcal{S}_i'(e)$ instead of $\mathcal{S} $. Under these assumptions, we will prove that our algorithm terminates in expected polynomial time. Recalling the proof of Theorem~\ref{main}, the latter  allows us to invoke  an  identical coupling argument and show that the price of making these assumptions is to increase the failure probability of our algorithm by an additive $1/n^{\gamma}$, where $\gamma = \gamma(\beta)$ can be made arbitrarily large by increasing $\beta$. This error probability can be subsumed by the expected running time of our algorithm.

For two flaws $f_{x_1}, f_{x_2}$, where $ x_1, x_2 \in V \cup E$, we consider the causality relation $f_{x_1} \sim f_{x_2} $ iff $\mathrm{dist}(x_1, x_2) \le t^2 + t+2$.  By inspecting procedure {\sc Fix}  it is not hard to verify that  this is a valid choice for a causality graph in the sense that no flaw $f$ can cause flaws outside $\Gamma(f)$. This is because, in order to determine whether a flaw $f_x$ is present in a state $\sigma$, we only need information about $\sigma$ in $G \cap S_{<t }(x) $, and  procedure {\sc FIX} locally modifies the state within a radius at most $t^2$ of the input subgraph $H$.

The algorithmic proof of Theorem~\ref{first_th}, which as we explained earlier is the key ingredient
in making Kahn's result constructive, follows almost immediately by combining Theorem~\ref{our_theorem} with  Lemma~\ref{termination_list} below, whose proof can be found in Section~\ref{proof_main_list_lemma}.

\begin{lemma}\label{termination_list}
Let $f \in \{f_e, f_v \}$   for   an edge $e$ and  a vertex $v$. Then:
\begin{align*}
\gamma_f \le \frac{1 }{ \Delta^{3(t^2 + t +2) } }  ,
\end{align*}
where the charges are computed with respect to measure $\mu$ and the algorithm that samples from the ideal distributions.
\end{lemma}

\begin{proof}[Constructive Proof of Theorem~\ref{first_th}]

Recall from~\eqref{symmetric} that, setting $x_f = \frac{1}{1+\max_{f \in F} | \Gamma(f) | }$ for each flaw $f$, condition~\eqref{generalAlgoLLL} with   $\epsilon = 1/4$  is implied by
\begin{align}\label{final_const}
\max_{f \in F} \gamma_f  \cdot \bigl(1 + \max_{f \in F} | \Gamma(f) | \bigr) \cdot  \mathrm{e} \le 3/4    .
\end{align}
Clearly, for each flaw $f$, $|\Gamma(f)| = O(  \Delta^{ 2(t^2+t+2) } ) $ so, by Lemma~\ref{termination_list}, condition~\eqref{final_const} is satisfied for all sufficiently large $\Delta$. 
Thus, Theorem~\ref{our_theorem} implies that, for every multigraph with large enough degree $\Delta_0$, the algorithm for each iteration terminates after an expected number of
\begin{align*}
O  \left(   (m+n)  \log_2 \left( \frac{1}{ 1-1/ \Delta^{ 2(t^2+t+2)  } } \right) \right) = O( n^2 ) 
\end{align*} 
steps.
\end{proof}
Finally, the proof of Theorem~\ref{main_list} is concluded by combining the algorithm for Theorem~\ref{first_th} with the greedy algorithm of Theorem~\ref{second_th}.  It remains only for us to prove
Lemma~\ref{termination_list} stated above.  This we do in the next subsection.

\subsubsection{Proof of Lemma~\ref{termination_list}.}\label{proof_main_list_lemma}

Let $\Omega_1 = \prod_{i  =1}^{N} \mathcal{M}(G_i)  $ and $  \Omega_ 2 =  \Omega_3 = \prod_{i =1 }^N \{ 0,1 \}^{ E_i} $. For notational convenience, sometimes we write $\Omega_1^i = \mathcal{M}(G_{i}) $ and $\Omega_2^i =  \Omega_3^i = \{0,1\}^{E_{i}}$, for $i \in [N]$.

 Let $\nu_1$ be the distribution over $\Omega_1$  induced by the product of distributions $\mu_i$, $ i \in [N]$. Let also  $\nu_2, \nu_3$ be the distributions over $\Omega_2$ and $\Omega_3$ induced by the product of activation and equalizing coin flips of each color $i \in [N]$, respectively. Recall that we can take a sample from $\mu$ by sampling from  $\nu_1 \times \nu_2 \times \nu_3$ to obtain a tuple  $\xi = (M,a,h) \in \Omega_1 \times \Omega_2 \times \Omega_3$, and then sample uniformly from an element from $\prod_{i=1}^N \prod_{e \in E(G_i) } \mathcal{S}_i'(e,\xi)$. Moreover, note that each $\nu_j$ is the product of $N$ distributions $\nu_j^i$, one for each color $i \in [N]$. For example, notice that $\nu_1^i $ is another name for $ \mu_{i}$, while $\nu_2^i$ is the product measure over the edges of $G_{i}$ induced by flipping a coin with probability $\alpha$ for each edge.

For  $\sigma_1 = (M_1, M_2, \ldots, M_N)  \in \Omega_1$,  a subgraph $H$,  and an integer $d >0$, we define the quantities $ Q_H(d,\sigma_1) = \left(M_1 - S_{<d}(H),  \ldots, M_N - S_{ <d}(H)  \right)$ and $Q_H^i(d,\sigma_1) = M_i - S_{< d } (H)$, similarly  to the proof of Lemma~\ref{termination}.  Moreover, for $\sigma_2 \in \Omega_2$ that represents the outcome of the activations,   we let $A_H(d, \sigma_2 )$ denote  the restriction of $\sigma_2$ to $M_i - S_{<d}(H) $ for each color $i \in [N]$.  In the same fashion, for $\sigma_3 \in \Omega_3$ that represents the outcome of the equalizing coin flips, we let $C_H(d, \sigma_3)$ denote the restriction of $\sigma_3 $ to $ M_i - S_{<d}(H) $ for each color $i \in [N]$.  For $\sigma_2 \in \Omega_2, \sigma_3 \in \Omega_3$, we also define $A_H^i(d,\sigma_2)$ and $C_H^i(d,\sigma_3)$, $i \in [N]$, similarly to $Q_H^i(d,\sigma_1)$. Finally, for $\xi= (\sigma_1, \sigma_2, \sigma_3) \in \Omega_1 \times \Omega_2 \times \Omega_3 $, define $R_H(d, \xi  ) = ( Q_H(d,\sigma_1), A_H(d,\sigma_2), C_H(d,\sigma_3) )  $.

Our goal will be to show that, for every $x \in V \cup E$,
\begin{align}\label{list_lemma_goal}
\gamma_{f_x}  = \max_{\tau \in \Omega } \mu(  \sigma \in f_x \mid R_x(t^2, \sigma)  =  R_x(t^2, \tau )  ) ,
\end{align}
where $\sigma$ is a random state distributed according to $\mu$. Note that combining~\eqref{list_lemma_goal} with Lemma~\ref{key_list_kahn_lemma_new} will conclude the proof of 
Lemma~\ref{termination_list}.

We only prove~\eqref{list_lemma_goal}   for  $f_e$-flaws, since the proof for $f_v$ flaws is very similar. Observe that whether flaw $f_e$ is present at a state $\sigma$ is determined by $\bigcup_{i=1}^N\left( G_i \cap G_{<t}(e) \right)$, the entries of the activation and equalizing flip vectors of each color $i \in [N]$ that correspond to edges in  $G_i \cap G_{<t}(e)$, and the value of the ``input randomness" strings $\{ s_i(e) \}_{i=1}^N$.  With that in mind,  for each color $i $  let $M_i(t,e) = M_i \cap E(G_i \cap G_{<t}(e) )$  and $a_i(t,e), h_i(t,e) $ denote the (random) vectors constraining   the entries of the activation and equalizing coin flip vectors for color $i$ that correspond to the edges of $G_{i} \cap G_{<t}(e)$. Let also $\mathcal{D}_i(t,e) $ denote the domain of possible values of $(M_i(t,e), a_i(t,e), h_i(t,e), s_i(e)  ) $.

The fact that we can determine whether $f_e$ is present in a state by examining local information around $e$  implies that there exists a set $X_e = X_e(t)$ of vectors of size $N$ such that the  $i$-th entry of a vector $x \in X_e$ is an element of $\mathcal{D}_{i}(t,e)$, and so that
\begin{align}\label{decomp_0}
f_e = \bigcup_{ x \in X_e }   \bigcap_{i  \in [N]  }  \left( \left(   M_{i}(t,e) ,a_{i}(t,e) , h_{i}(t,e),s_i(e)   \right) =  x_i \right).
\end{align}
For a state $\sigma \in \Omega '$, let $x_e^{\sigma}$ be the $N$-dimensional  random vector whose $i$-th entry is  $(M_{i}(t,e), a_{i}(t,e), h_{i}(t,e), s_i(e) ) $.
According to~\eqref{decomp_0}, for $\tau \in \Omega'$  we have
\begin{align}
\mu( \sigma \in f_e \mid R_e(t^2,\sigma)  = R_e(t^2,\tau))  =   \sum_{x \in X_e } \prod_{i=1}^N \mu( x_{e,i}^{\sigma} = x_i \mid R_e(t^2,\sigma)  = R_e(t^2,\tau))\label{decomp_1} ,
\end{align}
since the random choices of matching, activation, and equalizing coin flips for each color are independent.
For an $N$-dimensional vector $x$ whose $i$-th entry is an element of $\mathcal{D}_{i}(t,e)$,  we write $x_i(j)$ to denote the $j$-th element of tuple $x_i$. Thus, recalling the definition of the distributions $\nu_j^i$, we have 
\begin{align}\label{decomp_2}
\mu( x_{e,i}^{\sigma} = x_i \mid R_e(t^2,\sigma) =  R_e(t^2,\tau))  = \prod_{j =1}^3 \nu_j^i( x_{e,i}^{\sigma}(j )=  x_i(j)  \mid  R_e(t^2,\sigma)  = R_e(t^2,\tau) ) \cdot \frac{1}{ | \mathcal{S}_i'(e, \xi_i )  |   }  ,
\end{align}
where  $\xi_i = (x_i(1), x_i(2), x_i(3)  )$.

 Recall now that for a subgraph $H$, multigraph  $G_{<d+1}(H)$  is induced by $S_{<d+1}(H)$, and   $\mathcal{M}_{d+1}^i(H,\sigma) $  is the set of matchings of $G_{<d+1}(H) $ that are compatible with $Q_H^i(d,\sigma_1)$. Hence,
 \begin{align}
 &\nu_1^i( x_{e,i}^{\sigma}(1 )=  x_i(1)  \mid  R_e(t^2,\sigma)  = R_e(t^2,\tau) )  \nonumber \\
  &=\;  \nu_1^i( x_{e,i}^{\sigma}(1 )=  x_i(1)  \mid  Q_e^i(t^2,\sigma_1)  = Q_e^i(t^2,\tau_1) )  \nonumber \\
 & =\;  \frac{ \nu_1^i(( x_{e,i}^{\sigma}(1 )=  x_i(1) ) \cap  (Q_e^i(t^2,\sigma_1)  = Q_e^i(t^2,\tau_1) )  ) }{ \nu_1^i(  Q_e^i(t^2,\sigma_1)  = Q_e^i(t^2,\tau_1)    )   }  \nonumber  \\
& =\; \frac{ \sum_{ M \in \mathcal{M}_{t^2+1}^i(e,\tau_1), M \cap E(G_{<t}(e))  = x_{i}(1) }  \lambda_{i} (M) }{ \sum_{M \in \mathcal{M}_{t^2+1}^i(e,\tau_1) }  \lambda_{i}(M) }  \label{decomp_3}.
 \end{align}
Moreover, we clearly have 
\begin{align}
&\nu_2^i( x_{e,i}^{\sigma}(2 )=  x_i(2)  \mid  R_e(t^2,\sigma)  = R_e(t^2,\tau) ) \nonumber \\
 &\quad =  \nu_2^i( a_{i}(t,e) = x_i(2) ) \label{decomp_4}  ; \\
&\nu_3^i( x_{e,i}^{\sigma}(3 )=  x_i(3)  \mid  R_e(t^2,\sigma)  = R_e(t^2,\tau) )  \nonumber\\
  &\quad = \nu_3^i(  h_{i}(t,e) = x_i(3) )  \label{decomp_5}.
\end{align}
We will use~\eqref{decomp_1}-\eqref{decomp_5} to show that, for $\sigma$ distributed according to $ \mu $, and any state $\tau \in \Omega'$,
\begin{align}\label{goalaki}
\sum_{ \omega \in f_e} \frac{ \mu(\omega) }{  \mu(\tau) } \rho_{f_e}(\omega,\tau)  = \mu( \sigma \in f_e \mid R_e(t^2, \sigma) = R_e(t^2,\tau) ) .
\end{align}
According to the definition of $\gamma_{f_e} $, maximizing~\eqref{goalaki}  over $\tau \in \Omega'$ yields~\eqref{list_lemma_goal}.

To compute the sum in~\eqref{goalaki} we need to determine the set of states $\mathrm{In}_e(\tau) = \{  \omega: \rho_{f_e}(\omega, \tau) > 0  \}$. We claim  that for each $\omega  \in \mathrm{In}_e(\tau)$ we have that $R_e(t^2,\omega) = R_e(t^2,\tau)$. 

To see this,  let 
\begin{eqnarray*}
&\omega   =  (\omega_1, \omega_2, \omega_3, \omega_4)  \nonumber \\
 & =   \left( ( \omega_1^1, \ldots, \omega_1^N), (\omega_2^1, \ldots, \omega_2^N ), (\omega_3^1, \ldots \omega_3^N   ), \omega_4     \right)   ;\\
&\tau  =  (\tau_1, \tau_2, \tau_3,\tau_4) \nonumber  \\
&= \left( (\tau_1^1, \ldots, \tau_1^N ), (\tau_2^1, \ldots, \tau_2^N   ), (\tau_3^1, \ldots ,\tau_3^N ),\tau_4 \right) ,
\end{eqnarray*}
where $\omega_j, \tau_j \in \Omega_j$ and   $\omega_j^i, \tau_j^i \in \Omega_j^i$ for $j \in \{ 1,2,3\}$ and $\omega_4,\tau_4$ are tuples of input randomness strings. To express the probability distribution $\rho_{f_e}(\omega, \cdot)$ in a convenient way we consider the following  $3 N$ distributions. For each $i \in [N]$ we have a probability distribution $\rho_{f_e}^{i,1}(\omega_1^i, \cdot) $  corresponding to Line~\ref{first} of {\sc{ FIX}} and color $i$, and similarly, for $\omega_2^i, \omega_3^i$ we have probability distributions $\rho_{f_e}^{i,2}(\omega_2^i, \cdot), \rho_{f_e}^{i,3}( \omega_3^i, \cdot)$, corresponding to Lines~\ref{second},~\ref{third}  of {\sc{FIX}}  and color $i$, respectively. Recalling procedure {\sc{Resample}}, we see that the support of $\rho_{f_e}^{i,1}(\omega_1^i, \cdot )$ is  $\mathcal{M}_{t^2+1}^{i}(e,\omega_1 )$, and hence it must be the case that $Q_e^i(t^2, \omega_1) = Q_e^i(t^2, \tau_1)$ for every $i \in [N]$ and state $\omega \in \mathrm{In}_e(\tau) $. Similarly, by inspecting procedure {\sc FIX} one can verify that $A_e^i(t^2,\omega_2) = A_e^i(t^2,\tau_2) $ and that $C_e^i(t^2, \omega_3 ) = C_e^i( t^2,\tau_3) $ for each $i \in [N]$. Hence, $R_e(t^2,\omega) = R_e(t^2,\tau)$, as claimed.

For each $\omega \in f_e$,
\begin{align}\label{manip_1}
 \frac{  \mu(\omega) }{\mu(\tau)} \rho_{f_e}( \omega,\tau )  =    \prod_{i =1}^N \frac{  1 }{ |\mathcal{S}_i'(e,  \xi(\omega)   ) | }  \prod_{j =1}^3 \frac{  \nu_j^i(\omega_{j}^i) }{ \nu_j^i(\tau_j^i)} \rho_{f_e}^{i,j} ( \omega_j^i, \tau_j^i)    =:  \prod_{i =1}^N \frac{  1 }{ |\mathcal{S}_i'(e,  \xi(\omega)   ) | }  \prod_{j =1}^3 r_{i,j}(\omega),
\end{align}
since we have assumed that in Line~\ref{fourth} of { \sc FIX }  we update string $s_i(e)$ by sampling uniformly from  $\mathcal{S}_i'(e)$ instead of $\mathcal{S} $.

We will now give an alternative expression for each $r_{i,j}(\omega)$ in order to relate~\eqref{manip_1} to~\eqref{goalaki}. We start with $r_{i,1}(\omega)$. The fact that $Q_e^i(t^2,\omega_1) = Q_e^i(t^2, \tau_1)$ for each  $\omega \in \mathrm{In}_e(\tau) $ implies that 
\begin{align}\label{manip_2}
\frac{ \nu_1^i( \omega_1^i )}{ \nu_1^i( \tau_1^i) }  =  \frac{  \lambda_{i}( \omega_1^i  \cap E(G_{i,<t^2+1} (e)  )) }{  \lambda_{i}( \tau_1^i \cap E(G_{i,<t^2+1 }(e)  ) ) }.
\end{align}
To see this recall the definition of mutligraph $G_{i,<d+1}(H) $ in the text above the definition of procedure {\sc Resample}.

Furthermore, since we have assumed that the hard-core distribution in Lines~\ref{seven},~\ref{eight} of  {\sc{Resample}} is induced by the ideal vector of activities $\lambda_{i}$, we have 
\begin{align}\label{manip_3}
\rho_{f_e}(\omega_1^i, \tau_1^i)= \frac{  \lambda_{i}( \tau_1^i  \cap E(G_{i,<t^2+1}(e) )  }{ \sum_{ M \in \mathcal{M}_{t^2+1}^i(e,\omega_1)  } \lambda_{i}(M)  } .
\end{align}
Combining~\eqref{manip_2} with~\eqref{manip_3} and the fact that $Q_e^i(t^2,\omega_1) = Q_e^i(t^2, \tau_1)$ we obtain
\begin{align}
r_{i,1}(\omega) =     \frac{  \lambda_{i}( \omega_1^i  \cap E(G_{i,<t^2+1}(e) )  }{ \sum_{ M \in \mathcal{M}_{t^2+1}^i(e,\tau_1)  } \lambda_{i}(M)  } \label{manip_4} .
\end{align}
Recall now the definitions of $a_{i}(t,e)$ and $h_{i}(t,e)$. The fact that $A_e^i(t^2, \omega_2) = A_e^i(t^2, \tau_2) $  for each  $\omega \in \mathrm{In}_e(\tau) $ implies that
\begin{align}\label{manip_5}
\frac{ \nu_2^i( \omega_2^i )}{ \nu_2^i( \tau_2^i) }  = \frac{ \nu_2^i( a_{i}(t,e)   = x_{e,i}^{\omega}(2) )  }{\nu_2^i( a_{i}(t,e)   = x_{e,i}^{\tau}(2) )  }     .
\end{align}
Further, since in Line~\ref{second} of { \sc FIX } we simply flip a coin independently with success probability $\alpha$ for each edge of $G_{i, <t^2+1}(e) $, we have
\begin{align}\label{manip_6}
\rho_{f_e}(\omega_2^i, \tau_2^i) = \frac{\nu_2^i( a_{i}(t,e)   = x_{e,i}^{\tau}(2) ) }{  \sum_{a } \nu_2^i( a_{i}(t,e) = a )} ,
\end{align}
where the sum in the denominator ranges over all the possible values for $a_{i}(t,e)$. Thus, combining~\eqref{manip_5} with~\eqref{manip_6} we get
\begin{align}\label{manip_7}
r_{i,2}(\omega) =  \frac{\nu_2^i( a_{i}(t,e)   = x_{e,i}^{\omega}(2) ) }{  \sum_{a } \nu_2^i( a_{i}(t,e) = a )}  .
\end{align}
Finally, an identical argument shows that
\begin{align}\label{manip_8}
r_{i,3}(\omega) =  \frac{\nu_3^i( h_{i}(t,e)   = x_{e,i}^{\omega}(2) ) }{  \sum_{h} \nu_3^i( h_{i}(t,e) = h )},
\end{align}
where the sum in the denominator ranges over all the possible values for $h_i(t,e)$.

For $x \in X_e$, let  $\Omega_{e,x} = \{ \omega: x_e^{\omega} = x \}$ .  For $\sigma$ distributed according to $ \mu$, the left-hand side of~\eqref{goalaki} can be written as
\begin{eqnarray}
\sum_{x \in X_e} \sum_{ \omega \in \Omega_e}  \frac{ \mu( \omega) }{ \mu( \tau)  } \rho_{f_e}( \omega, \tau )    &=  &\sum_{x \in X_e} \sum_{ \omega \in \Omega_{e,x} }   \prod_{i=1}^N     \frac{  1 }{ |\mathcal{S}_i'(e,  \xi(\omega)   ) | }  \prod_{j=1}^3  r_{i,j}(\omega) \nonumber  \\
  & =& 	\sum_{x \in X_e} \prod_{i=1}^N \frac{1}{ |  \mathcal{S}_i'(e, (x_i(1), x_i(2), x_i(3) )  )   |    }    \prod_{j=1}^3 	\sum_{  \substack{ \omega  \in \Omega_{e,x} \\x_{e,i}^{\omega} = x_i(j) } } r_{i,j}(\omega)  \label{explain_1} \\
 & =&\sum_{x \in X_e} \prod_{i=1}^N \frac{1}{ |  \mathcal{S}_i'(e,  \xi_i)   |    }   \prod_{j=1}^3  \nu_j^i( x_{e,i}^{\sigma}(j)  = x_i(j) \mid R_e(t^2,\sigma) = R_e(t^2,\tau) ) \label{explain_2} \\
& = &\mu(\sigma \in f_e \mid R_e(t^2,\sigma) = R_e(t^2,\tau) )  \nonumber,
\end{eqnarray}
where $\xi_i =  (x_i(1), x_i(2), x_i(3) ) $, concluding the proof of~\eqref{goalaki}. Note that~\eqref{explain_2} follows by~\eqref{decomp_3} and~\eqref{manip_4} for $j=1$,~\eqref{decomp_4} and~\eqref{manip_7} for $j=2$, and~\eqref{decomp_5} and~\eqref{manip_8} for $j=3$.  This concludes the proof of \eqref{list_lemma_goal} and hence of Lemma~\ref{termination_list}

\section{Acknowledgements}

We thank Dimitris Achlioptas for many helpful discussions. We are also thankful to anonymous reviewers for helping us correct inaccuracies in previous versions of the paper and, in particular, for the idea to include the random bits used for approximating the edge-marginals in  the algorithm of Section~\ref{list_algo} in  the definition of its state space, which simplified our approach.

\bibliographystyle{plain}
\bibliography{Kolmo}

\appendix

\section{Proof of Lemma~\ref{basic_lemma}}\label{bounds}

We will need the following standard concentration bound (see, e.g.,~\cite[Section 10.1]{mike_book}).

\begin{lemma}\label{simple_bound}
Let $X$ be a random variable determined by $n$ independent trials $T_1, \ldots, T_n$, and such that changing the outcome of any one trial can affect $X$ by at most $c$. Then
\begin{align*}
\Pr(| X - \ex[X] | > \lambda  ) \le 2 \mathrm{e}^{- \frac{\lambda^2}{ 2c^2 n}  } .
\end{align*}
\end{lemma}

\begin{proof}[Proof of Part~(a) of Lemma~\ref{basic_lemma}]
Recall that $t =8(K+1)^2 \delta^{-1} + 2$ and that $\delta = \frac{\epsilon}{4}$. Consider a random state $\sigma $ distributed according to  $\mu$ and a fixed state $\tau \in \Omega$, and  notice that applying Theorem~\ref{hardcore_dist} with the parameter $\epsilon$ instantiated to~$\delta$  and our choice of $t$ imply that  
\begin{align*}
\mu( e \in M_i \mid Q_v^i(t, \sigma) = Q_v^i(t, \tau)   ) & \ge (1- \delta) \frac{1-\delta}{ \chi_e^*(G) }  \ge \frac{1 - \frac{\epsilon}{2} }{\chi_e^*(G) } ,
\end{align*}
for any vertex $v$, any edge $e$ incident on~$v$ and any $i \in [N]$. This implies
\begin{align*}
\ex[ d_{G_{\sigma} } (v)  \mid    Q_v^i(t, \sigma) = Q_v^i(t, \tau)      ] \le \Delta  \left( 1 - \frac{1 - \frac{\epsilon}{2}}{ \chi_e^*(G) }  \right)^N   \le \chi_e^*(G)  \left( 1 - \frac{1 - \frac{\epsilon}{2}}{ \chi_e^*(G) }  \right)^N.
\end{align*}
\blue{
Now, recalling that  $N = \lfloor  \chi_e^*(G)^{ \frac{3}{4} }  \rfloor  \sim \Delta^{3/4} $ and $\epsilon \le \frac{1}{10}$, for sufficiently large $\Delta$ we have
\begin{align}\label{third_cond_bound}
\ex[ d_{G_{\sigma} } (v) \mid Q_v^i(t, \sigma) = Q_v^i(t, \tau)     ]  \le \chi_e^*(G) \left( 1 - (1- \epsilon^3 ) \frac{ (1- \frac{\epsilon}{2} ) N}{ \chi_e^*(G) }   \right)  \le \chi_e^*(G) - \left( 1- \frac{9\epsilon}{17} \right) N.
\end{align}
}
Further, since $c^* = \chi_e^*(G) - (1+\epsilon)^{-1} N$ and $\epsilon \le \frac{1}{10}$,~\eqref{third_cond_bound} yields
\begin{align*}
\ex[ d_{G_{\sigma} } (v) \mid Q_v^i(t, \sigma) = Q_v^i(t, \tau)     ]    \le c^* - \left( 1 - \frac{9 \epsilon}{17 } - (1+\epsilon)^{-1}  \right) N  \le c^* - \frac{\epsilon}{3} N.
\end{align*}
As the choices of the $M_i$ are independent and each affects the degree of $v$ in $G'$ by at most $1$, we can apply Lemma~\ref{simple_bound} with $\lambda = ( \frac{\epsilon }{ 3} - \frac{ \epsilon}{ 4}  ) N  =  \frac{ \epsilon}{12 } N $ to prove part (a). In particular,  since $N = \lfloor  \chi_e^*(G)^{ \frac{3}{4} }  \rfloor  \sim \Delta^{3/4}$ we have 
\begin{align*}
\mu\left( d_{G_{\sigma}(v)} > c^* - \frac{\epsilon}{4} N  \;\Big|\; Q_v^i(t, \sigma) = Q_v^i(t, \tau)     \right)  \le\;  2 \mathrm{e}^{ -\frac{ \lambda^2}{ 2 N }   } \; \le\; \frac{ 1} {  \Delta^{C +  2 \Delta^{\frac{1}{3} }  } }   ,
\end{align*}
for any constant $C $ for sufficiently large $\Delta$.
\end{proof}

\begin{proof}[Proof of Part (b) of Lemma~\ref{basic_lemma}]
The proof of part (b) is similar. Consider again a random state $\sigma$ distributed according to $\mu$ and fix a state $\tau \in \Omega$. Theorem~\ref{hardcore_dist} implies that  for each $i  \in [N ]$, the probability that an edge $e$ with both endpoints in $H$
is in $M_i$, conditional on  $Q_H^i(t, \sigma ) = Q_H^i(t, \tau ) $, is at least $( 1- \delta)  \frac{1-\delta}{ \chi_e^*(G) } \ge \frac{1- \frac{\epsilon}{2} }{ \chi_e^*(G) } $. Moreover, Edmonds' characterization of the matching polytope (which we have already seen in the the proof of Lemma~\ref{sufficient_flaws}) implies that  the number of edges in $G$ with both endpoints in $H$ is at most  $\chi_e^*(G)  \lfloor \frac{ V(H) -1 }{ 2}   \rfloor$.
Similar calculations to those in part~(a) reveal that 
\begin{align*}
\ex[  |E_{\sigma}(H)|  \mid Q_H^i(t,\sigma) = Q_H^i(t,\tau) ] \le \left( \frac{V(H)-1 }{ 2} \right)  (c^* - \frac{\epsilon}{3} N  ),
\end{align*}
where $E_{\sigma}(H)$ is the set of edges of $G_{\sigma}$ induced by $H$. Since the choices of matchings $M_i$ are independent and each affects $|E_{\sigma}(H)|$ by at most $\frac{|V(H)| -1 }{ 2} $, we can again apply Lemma~\ref{simple_bound} to prove part (b).
\end{proof}

\end{document}